\documentclass[]{article}

\usepackage[margin=1in]{geometry}

\usepackage{graphicx}
\usepackage{enumerate}

\usepackage{amsmath,amsfonts,amssymb,amsthm}
\usepackage{bbm}
\usepackage{hyperref}
\usepackage{xcolor}
\usepackage{multirow}
\usepackage[justification=centering]{subcaption}

\usepackage[numbers]{natbib}
\usepackage{diagbox}
\usepackage{tikz}
\usetikzlibrary{arrows.meta}
\usepackage{pgfplots}
\pgfplotsset{compat=1.15}

\newtheorem{theorem}{Theorem}[section]

\newtheorem{prop}[theorem]{Proposition}

\newtheorem{claim}[theorem]{Claim}

\newtheorem*{remark*}{Remark}
\newtheorem*{notation*}{Notation}
\newtheorem*{observation*}{Observation}
\newtheorem*{theorem*}{Theorem}

\newcommand{\mbs}[1]{\text{\textbf{#1}}}
\newcommand\floor[1]{\lfloor#1\rfloor}
\newcommand\ceil[1]{\lceil#1\rceil}
\newcommand{\twopartdef}[4] {
	\left\{
		\begin{array}{lll}
			#1 & \mbox{if } #2 \\
			#3 & \mbox{if } #4
		\end{array}
	\right.
}

\begin{document}
\title{Risk-Free Bidding in Complement-Free Combinatorial Auctions}

\author{
Vishnu V. Narayan\\
McGill University\\
\texttt{vishnu.narayan@mail.mcgill.ca}
\and
Gautam Rayaprolu\\
McGill University\\
\texttt{gautam.rayaprolu@mail.mcgill.ca}
\and
Adrian Vetta\\
McGill University\\
\texttt{adrian.vetta@mcgill.ca}
}



\date{\today}

\maketitle

\begin{abstract}
We study risk-free bidding strategies in combinatorial auctions with incomplete information. Specifically, what is the maximum profit that a complement-free (subadditive) bidder can guarantee in a multi-item combinatorial auction?
Suppose there are $n$ bidders and $B_i$ is the value that bidder $i$ has for the entire set of items. We study the above problem from the perspective of the first bidder, Bidder~1. In this setting, the worst case profit guarantees arise in a duopsony, that is when $n=2$, so this problem then corresponds to playing an auction against a 
budgeted adversary with budget $B_2$. We present worst-case guarantees for two simple and widely-studied combinatorial auctions; namely, the sequential and simultaneous auctions, for both the first-price and second-price case. In the general case of distinct items, our main results are for the class of {\em fractionally subadditive} (XOS) bidders, where we show that for both first-price and second-price sequential auctions Bidder~$1$ has a strategy that guarantees a profit of at least $(\sqrt{B_1}-\sqrt{B_2})^2$ when $B_2 \leq B_1$, and this bound is tight. More profitable guarantees can be obtained for simultaneous auctions, where in the first-price case, Bidder~$1$ has a strategy that guarantees a profit of at least $\frac{(B_1-B_2)^2}{2B_1}$, and in the second-price case, a bound of $B_1-B_2$ is achievable. We also consider the special case of sequential auctions with identical items, for which we provide tight guarantees for bidders with subadditive valuations.
\end{abstract}

\section{Introduction} \label{s:intro}
What strategy should a bidder use in a combinatorial auction for a collection $I$ of items?
This paper studies this question for sequential and simultaneous auctions.
To motivate this question and to formalize the resultant problem, we begin with sequential auctions, which are perhaps
the simplest and most natural method by which to sell multiple items.
These auctions, where the items are ordered and sold one after another, are commonplace in auction house and online sale
environments. The inherent simplicity of a sequential auction arises because a standard single-item mechanism, 
such as an {\em ascending-price}, {\em first-price}, or {\em second-price auction}, can be used for each item in the collection. 
But there is a catch! Whilst a single-item auction is very well understood from both a theoretical perspective -- see, for example, 
the seminal works of \citet{Vic61} and \citet{Mye81} -- and a practical perspective, the concatenation 
of single-item auctions is not.

From a bidder's viewpoint, sequential auctions are perplexing for a variety of reasons. 
To understand this, observe that a sequential auction can be modelled as an \textit{extensive form game}. In such games the basic
notion of equilibrium is a \textit{subgame perfect equilibrium}. 
Unfortunately, these equilibria are, in general, hard to 
compute; see, for example, \cite{DFP06, PST12b,Ete19}. Intriguing structural properties can be derived for the equilibria of sequential auctions, but the recursive nature
of this structure makes reasoning about equilibria complex. 
For example, in a two-bidder sequential auction each subgame reduces to a standard {\em auction with independent valuations} (see \citet{GS01}), but multi-bidder sequential auctions are more complex. Specifically, \citet{PST12} showed that in a first-price sequential auction, each subgame reduces to an {\em auction with externalities}. \citet{NPV19} provided an analogous result for second-price sequential auctions, and we refer the reader to this text for a detailed exposition of equilibria in multi-bidder sequential auctions. 

Subgame perfect equilibria suffer from an additional drawback in that they are very sensitive to changes in the environment.
In particular, the equilibria in a game can change significantly with small changes to the payoff values. It follows that prescriptions (if such prescriptions can be computed) 
derived from the complete information setting, as studied in the aforementioned papers, are unlikely to extend reliably to more practical settings with incomplete information. Given the overwhelming complexity and informational sensitivity of a subgame perfect equilibrium, the following question immediately arises: {\em As a bidder in a sequential auction, how can you (approximately) play your equilibrium strategy?} However, we would argue that this is the wrong question to ask. 
Given that subgame perfect equilibria are very complex and informationally sensitive, it is extremely unlikely that the other bidders will be able to play their 
equilibrium strategies. Then, as a participant in a sequential auction, {\em what bidding strategy should you use instead?} Similar computational and informational motivations also arise for the case of simultaneous auctions. In this paper, we consider the above question for both types of auctions.

Evidently, the answer to this question depends upon the objective of the bidders, their computational resources, and the informational structure inherent
in the auction. We study this problem from the perspective of Bidder~$1$ in the following very general incomplete information setting. What is the maximum {\em risk-free profit} that Bidder~$1$ can guarantee for herself in a combinatorial auction? 
Here, Bidder~$1$ knows her own entire valuation function but does not know the valuation function of the other agents. Clearly, if the other agents' bids are unrestricted then no guarantee is possible. Consequently, we impose a mild assumption on the other agents: Bidder~$i$ can spend at most some fixed budget $B_i$ over the course of the multi-item auction. We will see that the critical case to 
analyze is when there are just two bidders (Bidder ~$1$ and Bidder~$2$). We also assume that the only information Bidder~$1$ has on the other bidder is an estimate that his value for 
the entire collection of items is at most $B_2$; beyond this trivial upper bound, she has no specific information on the values the other bidder has for any subset of 
the items. In the worst case, to maximize her guaranteed profit, we can model this problem as Bidder~$1$ competing 
in the auction against a 
single \textit{adversary} who is incentivized to keep Bidder~$1$'s utility low, and is willing to spend at most his budget $B_2$. This type of approach is analogous to that of a {\em safety strategy} in bimatrix games.
This worst-case setting then corresponds to a special case of an auction with externalities, where Bidder~2 has no value for the items themselves, but is willing to bid on the items simply to prevent Bidder~1 from acquiring them.
In this paper, we will then quantify the maximum risk-free profitability when the valuation function of Bidder~$1$ belongs to the class of subadditive (complement-free)
functions and its subclasses. Interestingly, given the valuation class, tight bounds can be obtained that depend only on $B_1$ (the value Bidder~1 has for the entire set of items) and $B_2$. For example, the risk-free profitability of the class of fractionally subadditive~(XOS) valuation functions is $(\sqrt{B_1}-\sqrt{B_2})^2$, for $B_2\le B_1$,
and this bound is tight. Similarly, we present tight (but more complex) bounds for the class of subadditive valuation functions when the items are identical. We also analyze simultaneous auctions, for which we show that the risk-free profitability of the XOS class is at least
$\frac{(B_1-B_2)^2}{2B_1}$ and $(B_1-B_2)$ for first-price and second-price auctions, respectively.

\subsection{Related Literature}
There is an extensive literature on sequential auctions. The study of incomplete information games was initiated by
Milgrom and Weber \cite{MW82, Web83}. Theoretical studies on equilibria in complete information games include \cite{GS01,PST12, NPV19}.
Given the abundance of sequential auctions in practice, there is also a very large empirical literature covering an assortment of applications ranging from 
antiques~\cite{GO07} to wine~\cite{Ash89} and from fish~\cite{GGK11} to jewellery~\cite{CGV96}.

Recently there has been a strong focus in the computer science community on the design of simple mechanisms.
For combinatorial auctions, {\em simultaneous auctions} are a notable example. These auctions are simple in that,
as with a sequential auction, a standard single-item auction mechanism is used to sell each item.
But in contrast, as the nomenclature suggests, these auctions are now held simultaneously rather than sequentially.
Two important streams of research in this area concern the price of anarchy in simultaneous auctions (see, for example,
\cite{CKS16, BR11, HKM11, FFG13})
and the hardness of computing an equilibrium (see \cite{CP14}).

There has also been a range of papers examining the welfare of equilibria in sequential auctions. 
\citet{BBB09} consider the case of identical items and show that equilibria provide
a factor $1-\frac{1}{e}$ approximation guarantee if there are two bidders with non-decreasing marginal valuations. 
\citet{PST12} study the case of multi-bidder auctions. For sequential first-price auctions, they prove a 
factor $2$ approximation guarantee for unit-demand bidders. In contrast, they show that equilibria can have arbitrarily 
poor welfare guarantees for bidders with submodular valuations. \citet{FLS13} extend this result to the case where each bidder has either a unit-demand or additive valuation function. 

Partly because of these negative results, a common assumption is that sequential auctions may not be a good 
mechanism by which to sell a collection of items. However, there are reasons to believe that, in practice, sequential auctions 
have the potential to proffer high welfare. For example, consider the influential paper of 
\citet{LLN06}. There, they present a simple greedy allocation mechanism with a factor $2$ welfare guarantee for allocating items 
to agents with submodular valuation functions. One interesting implication of this result is that if the items are sold via a second-price 
sequential auction {\em and} every agent (assuming submodular valuations) truthfully bids their marginal value in each round then the outcome will have at least half 
the optimal social welfare.

The study of auctions with \textit{externalities} was initiated by \citet{Fun96} and by \citet{JM96}. In an auction with externalities, each agent $i$'s value for an allocation of the items does not depend only on the items received by agent $i$, but on the identities of the agents that receive the other items. Examples of such settings include firms competing in a market that may seek to acquire items that are valuable to a competitor, or sports teams competing to sign a star athlete to prevent a rival team from signing them instead. In this work, the worst-case setting that we consider corresponds to a special case of an auction with externalities -- one where Bidder~2 has no value for the items themselves but is willing to bid on them in order to prevent Bidder~1 from acquiring them. Interestingly, one of the first papers to study auctions with externalities considers a similar special case. Specifically, \citet{JMS96} consider the sale of a single good to a set of agents with externalities, and construct a revenue-maximizing auction for the seller that extracts a surplus from the agents who do not obtain the auctioned item. The motivation for their work was the situation that occurred after the breakup of the Soviet Union, when Ukraine inherited a large nuclear arsenal. Despite having no direct interest in acquiring the weapons themselves, both the Unites States and Russia agreed to pay over US\$ 1 billion to Ukraine in exchange for the dismantling of the weapons in order to avoid the danger of proliferation.

\subsection{Overview and Results} \label{ss:results}
In Section~\ref{s:model} we explain the sequential auction model and give necessary definitions. We present our measure, the risk-free profitability
of a bidder in incomplete information multi-bidder auctions, and explain how to quantify it via a two-bidder adversarial sequential auction.
In Section~\ref{s:example} we present a simple sequential auction example ({\em uniform additive auctions}) to motivate the problem
and to illustrate the difficulties that arise in designing risk-free bidding strategies, even in very small sequential auctions with at most three items. 

Section~\ref{s:xos-bounds} and Section~\ref{s:subadd-bounds} contain our main results.
In Section~\ref{s:xos-bounds} we begin by presenting tight upper and lower bounds on the risk-free profitability of a fractionally subadditive (XOS) bidder.
For the lower bound, in Section~\ref{ss:xos-lb} we exhibit a bidding strategy 
that guarantees Bidder~$1$ a profit of at least $(\sqrt{B_1}-\sqrt{B_2})^2$.

In Section~\ref{ss:xos-ub} we describe a sequence of sequential auctions that provide an upper bound that is asymptotically equal to the aforementioned 
lower bound as the number of items increases. We prove these bounds for {\em first-price} sequential auctions, but nearly identical proofs show the bounds also apply
for {\em second-price} sequential auctions.
Next we prove that the risk-free profitability of an XOS bidder is lower in sequential auctions than in
simultaneous auctions. Equivalently, a 
budgeted adversary is stronger in a sequential auction than in the corresponding simultaneous auction. 
Specifically, in Section~\ref{ss:simultaneous}, we prove that an XOS bidder has a risk-free profitability of 
at least $\frac{(B_1-B_2)^2}{2B_1}$ in a first-price simultaneous auction and of at least $B_1-B_2$ in a second-price simultaneous auction.
Several other interesting observations arise from these results. First, unlike for sequential auctions, the power of the adversary differs in a simultaneous 
auction depending on whether a first-price or second-price mechanism is used: the adversary is stronger in a first-price auction.
Second, the risk-free strategies we present for simultaneous auctions require {\em no} information about the adversary at all. The performance of the strategy (its risk-free profitability) 
is a function of $B_2$, but the strategy itself does not require that Bidder~$1$ have knowledge of $B_2$ (nor an estimate of it).
Third, for the case of first-price simultaneous auctions, it is necessary that Bidder~$1$ use randomization in her risk-free strategy.

Finally, in Section~\ref{s:subadd-bounds} we study the risk-free profitability of a bidder with a subadditive valuation function. We give a 
possible explanation for why simple strategies fail to perform well in the general case. We then examine the special case where the items 
are identical. We derive tight lower and upper bounds for this setting.

\section{The Model} \label{s:model}

In this section we present our model. In particular, Sections~\ref{ss:seqmodel} and \ref{ss:simulmodel} describe the properties of sequential and simultaneous auction games (and the associated valuation functions of the agents) as used in practice and described in the related literature. Section~\ref{ss:mainmodel} describes the model of the zero-sum game that we will analyze to obtain our results.

\subsection{Sequential Auctions and Valuation Functions} \label{ss:seqmodel}
There are $n$ bidders and a collection $I = \{a_1,\ldots,a_m\}$ of $m$ items to be sold using a sequential auction.
In the $\ell$th round of the auction item $a_{\ell}$ is sold via a first-price (or second-price) auction.
We view the auction from the perspective of Bidder~$1$ who has a publicly-known valuation function $v_1 : 2^I \rightarrow \mathbb{R}_{\geq0}$ 
that assigns a non-negative value 
to every subset of items. We denote $v_1$ by $v$ where no confusion arises. This valuation function is assumed to 
satisfy $v(\emptyset) = 0$ and to be {\em monotone}, that is, $v(S) \leq v(T)$, 
for all $S \subseteq T$. When all the items 
have been auctioned, the \textit{utility} or \textit{profit} $\pi_1$ of 
Bidder~1 is her value for the set of items she was allocated minus the sum of prices of these items.

The sequential auction setting is captured by extensive form games. A {\em strategy} for player $i$ is a function that 
assigns a bid $b_i^t$ for the item $a_t$, depending on the previous bids $\{b_i^\tau\}_{i,\tau<t}$ of all players (and the allocation of the 
first $t-1$ items). The utility (profit) of a strategy profile $\mbs{b}$ for Bidder~1 is the profit Bidder~1 obtains when all bidders bid according to $\mbs{b}$.

The question we then study is how much profit Bidder~1 can guarantee herself. 
We examine the case where $v$ is in the class of {\em subadditive} or {\em complement-free} valuation functions.
Belonging to this class, of particular interest in this paper are {\em additive} functions,
{\em submodular} functions, and {\em fractionally subadditive} or {\em XOS} functions.
These functions are defined as follows.
\begin{itemize}
\item {\tt Subadditive (Complement-Free).} A function $v$ is subadditive if $v(S\cup T)\le  v(S) + v(T)$ for all $S,T \subseteq I$.
\item {\tt Additive (Linear).} A function $v$ is additive if $v(S) = \sum_{a\in S} v(a)$ for each $S \subseteq I$.
\item {\tt Submodular (Decreasing Marginal Valuations).} A function $v$ is submodular if $v(S\cup T)+ v(S\cap T)\le  v(S) + v(T)$ for all $S,T \subseteq I$.
\item {\tt Fractionally Subadditive (XOS).} A function $v$ is fractionally subadditive if there exists a 
nonempty collection of additive functions $\{\gamma_1,\ldots,\gamma_\ell\}$ on $I$ such that for 
every $S \subseteq I$, $v(S) = \max_{j\in[\ell]} \gamma_j(S)$.\footnote{This is the standard definition of XOS functions. Fractionally subadditive functions 
are defined in terms of fractional set covers; the equivalence between fractionally subadditive and XOS functions was shown by \citet{Fei09}.}
\end{itemize}

\noindent\citet{LLN06} showed that these valuation classes form the following hierarchy:

{\small 
{\sc Additive} $\ \subseteq\ $ {\sc Submodular} $\ \subseteq\ $ {\sc Fractionally Subadditive} $\ \subseteq\ $ {\sc Subadditive}}

Other important classes in this hierarchy include unit-demand and gross substitutes valuation functions, but they will not be
needed here.

\subsection{Simultaneous Auctions} \label{ss:simulmodel}

The simultaneous auction setting is similar to the sequential auction setting in that there are $n$ bidders and a collection $I = \{a_1,\ldots,a_m\}$ of $m$ items to be sold. Each bidder makes a bid on each item. Unlike the sequential case, there is now a single time period. Thus each bidder makes a vector of $m$ bids -- one for each item -- simultaneously.

\subsection{Bidding against an Adversary} \label{ss:mainmodel}
To quantify the maximum profit that Bidder~1 can obtain, without loss of generality, 
we may \textit{normalize} the valuation function (and corresponding auction) by scaling the values so that $v(I)=v_1(I) = 1$. Now the maximum guaranteed profit will depend on the strength of the other bidders. We quantify this by a parameter $B$: in the setting where each player $j \geq 2$ has valuation function $v_j$, $B$ is the sum of the total values of the other bidders, i.e., $B = \sum_{j=2}^n v_j(I)$.
This corresponds to an incomplete information auction where the only common knowledge are upper bounds on the value each agent has for the entire set of items. From the perspective of Bidder~$1$, it is apparent that the worst case arises when 
$n=2$, and so $B=B_2=v_2(I)$. Thus we may assume that $n=2$, and we can view Bidder~$1$ as playing against an {\em adversary} with a budget $B$. To see this, observe that for a fixed $B = \sum_{j=2}^n v_j(I)$ if there are $n>=3$ bidders then the worst case for Bidder~1 arises when the other bidders coordinate to act as a single adversary: however, when the budget is split between two or more other bidders then their ability to buy a single item of high value decreases.

Thus, Bidder~$1$ seeks a strategy that works well against a rational bidder who, by monotonicity, has a value at most $B$ for
any subset of the items. We model the game that we will analyze as a zero-sum game with the following properties. Bidder~$1$'s payoff in this game is simply her profit $\pi_1$ from the auctions, that is, her value for the set that she is allocated minus the sum of the prices of the items. Since Bidder~$2$ is viewed as an adversary, and the game is zero-sum, his payoff in this game, $\pi_2$, is equal to $-\pi_1$. Bidder~$2$ can evaluate his payoff during the game, i.e., Bidder~$1$'s valuation function is common knowledge. Additionally, here the adversary's budget constraint is tight: for example, in the sequential case, in time step $t$, if Bidder~$2$ paid $p_2^{t-1}$ for the items that have already sold, then his next bid $b_2^t$ is at 
most $B - p_2^{t-1}$. We call this the {\em risk-free sequential auction game} $\mathcal{R}(v,B)$. The \textit{guaranteed profit} for Bidder~1 is the minimum profit obtainable by playing a safety strategy in this game (i.e. the {\em value} of this game). For any normalized valuation $v$, we denote this profit by $\pi_1^*(v,B)$, or simply $\pi_1^*$ where there is no ambiguity. For any class of set functions $\mathcal{C}$ and any budget $B\in(0,1)$, we want to find the maximum profit Bidder~1 can guarantee in {\em all} $m$-item instances $\mathcal{R}(v,B)$ where $v\in\mathcal{C}$, which is precisely $\min_{v\in\mathcal{C}}\pi_1^*(v,B)$. We call this the \textit{risk-free profitability} $\mathcal{P}(\mathcal{C},B)$ of the class $\mathcal{C}$. For a 
budgeted adversary in a simultaneous auction,
the analogue of this budget-constrained bidding is that the {\em sum} of the adversary's bids on the items is at most the budget $B$. We define risk-free profitability analogously for simultaneous auctions. The focus of this paper is to quantify the risk-free profitability of the aforementioned classes of valuation functions for both sequential and simultaneous auction mechanisms. We note that for the sequential setting, in the course of our proofs we will show implicitly that the aforementioned zero sum game has a \textit{subgame perfect equilibrium} in pure strategies. Consequently this is the solution concept that we will analyze for our main result.

Table~\ref{tab:results} summarizes our obtained bounds by auction type, where the valuation functions of the agents are normalized as described previously.

\begin{table}[ht]
    \centering
    \begin{tabular}{ |c|c|c|  }
        \hline
         Valuation Class & Lower Bound & Upper Bound \\
        \hline
        \multicolumn{3}{|c|}{Sequential Auctions (First- and Second-Price)} \\
        \hline
        Additive & $(1-\sqrt{B})^2$ & $(1-\sqrt{B})^2$ \\
        Submodular & $(1-\sqrt{B})^2$ & $(1-\sqrt{B})^2$ \\
        XOS & $(1-\sqrt{B})^2$ & $(1-\sqrt{B})^2$ \\
        Subadditive (Identical) & $t^{*}(B)$ & \begin{tabular}{c}
            $t^{*}(B)$, $B\in(0,\frac{1}{4})$ \\ $(1-\sqrt{B})^2$, $B\in[\frac{1}{4},1)$
        \end{tabular} \\
        \hline
        \multicolumn{3}{|c|}{Simultaneous Auctions (First-Price)} \\
        \hline
        XOS & \begin{tabular}{c}
            $1-B^2$, $B\in(0,3-2\sqrt{2})$ \\ $\frac{(1-B)^2}{2}$, $B\in[3-2\sqrt{2},1)$
        \end{tabular} & \begin{tabular}{c}
            $1-2B$, $B\in(0,\frac{1}{4})$ \\ $\frac{2}{3}(1-B)$, $B\in[\frac{1}{4},1)$
        \end{tabular} \\
        \hline
        \multicolumn{3}{|c|}{Simultaneous Auctions (Second-Price)} \\
        \hline
        XOS & $(1-B)$ & $(1-B)$ \\
        \hline
        \end{tabular}
    \caption{Valuation Classes and their Risk-Free Profitability}
    \label{tab:results}
\end{table}

\section{Example: Uniform Additive Auctions} \label{s:example}
We now present a simple example of a sequential auction with an agent (Bidder~1) that strategizes against an adversary (Bidder~2), which will be helpful for two reasons. 
First, it illustrates some of the strategic issues facing the agent and, implicitly, the adversary in a sequential auction.
Second, these examples form base cases in our proof in Section~\ref{ss:xos-ub}.

The auction is defined as follows. Bidder~1 has an additive valuation function where each item has exactly the
same value. That is, for an auction with $m$ items, we have that $v(a_t)=\frac{1}{m}$. The adversary Bidder~2 has a budget $B$.
We call this the {\em uniform additive auction} on $m$ items and denote it by $\mathcal{A}_m$.
For our example, we are interested in uniform additive auctions where $m\le 3$. We denote by $b_i^j$ Bidder~$i$'s bid on item $j$.\medskip

\noindent{\tt One Item.}
First, consider the case $\mathcal{A}_1$. We have a single item $a_1$ with $v(\{a_1\}) = 1$ for Bidder~1. Let 
$b_1$ and $b_2$ be the bids placed on the item by Bidders 1 and 2, respectively. 
Clearly if $b_1 < B$ then the adversary's best response is to bid $b_2=b_1^+$ and win the item.
Then $\pi_1 = \pi_2 = 0$. On the other hand, if $b_1 \geq B$, then the adversary is constrained by his budget 
and cannot beat Bidder~1. Thus Bidder~1 wins and obtains a profit of $\pi_1=1-b_1$.
It follows that Bidder~1's risk-free strategy is to bid $B$ and win the item at price $B$ for a
guaranteed profit of $\pi_1^* = 1-B$. Clearly, if $B \geq 1$ then $\pi_1^* = 0$ since the adversary can 
prevent Bidder~1 from winning the item. Specifically, we have shown
    \begin{alignat}{3} \label{ex:oneitem}
    \pi_1^* &= \twopartdef{1-B}{0 \leq B < 1}{0}{1 \leq B}.&
\end{alignat}
\noindent{\tt Two Items.}
Now consider the case $\mathcal{A}_2$. So there are two items $a_1$ and $a_2$ 
and Bidder~1 has an additive valuation function with $v(\{a_1\}) = v(\{a_2\}) = \frac{1}{2}$ and $v(\{a_1,a_2\}) = 1$. 
We divide our analysis into three cases.
\begin{itemize}
    \item $B < \frac{1}{4}$: If $B < \frac{1}{4}$, then 
    Bidder~1 can bid $B$ on each item and win both items at price $B$ each, so her guaranteed profit is 
    at least $1-2B>\frac{1}{2}$. If Bidder~1 bids less than $B$ on 
    either item, then Bidder~2 can win that item, ensuring that Bidder~1's profit is less than her value of the other item, that is $\frac{1}{2}$. 
    Bidder~1's risk-free strategy is thus to bid $B$ on both items for a profit $\pi_1^* = 1-2B$.
    
    \item $\frac{1}{4} \leq B < \frac{1}{2}$: If Bidder~1 
    bids $b_1^1 = d$ on $a_1$, with $0 \leq d \leq \frac{1}{2}$, then Bidder~2 can either win by 
    bidding $b_2^1>d$ or lose by bidding $b_2^1<d$ (for now, we assume $d < B$). In the former 
    case, the adversary's budget in the second auction is $B-b_2^1$, and there is only one item 
    remaining. Then, reasoning as we did for the one item setting, it is easy to see that Bidder~1's 
    profit from the second item is $\pi_1 = \frac{1}{2} - (B-b_2^1) = \frac{1}{2}-B+b_2^1$.
    This is minimized (with value $\frac{1}{2}-B+d$) when Bidder~2 bids an amount negligibly 
    larger than $d$. In the latter case, the adversary loses the first item, so he has budget $B$ in the 
    second auction. Bidder~1's combined profit (on both items) is then $\pi_1 = (\frac{1}{2}-d) + (\frac{1}{2}-B) = 1-B-d$. 
    For $d=0$ we have $\frac{1}{2}-B+d < 1-B-d$ and for $d=B$ we have $\frac{1}{2}-B+d \ge 1-B-d$, since $B\ge \frac14$. 
    But $\frac{1}{2}-B+d$ is 
    increasing in $d$ and $1-B-d$ is decreasing in $d$. Therefore, assuming Bidder~2 plays a best response, 
    we see that $\pi_1$ is maximized 
    when the minimum of these values is maximized.
    That is $\pi_1^* = \max_{0\le d<B}\, \min \left[\frac12-B+d, 1-B-d\right]$. The optimal choice is $d = \frac{1}{4}$ giving $\pi_1^* = \frac{3}{4}-B$. Note that our assumption that $d<B$ is validated: 
    if Bidder~1 bids an amount $d$ that is greater than or equal to $B$ 
    on the first item the she will win both items for a total profit  $(\frac{1}{2}-d)+(\frac{1}{2}-B)= 1-d-B \leq 1-2B \leq \frac{3}{4}-B$.
    
    \item $\frac{1}{2} \leq B < 1$: Suppose Bidder~1 bids $b_1^1 = d$ on $a_1$, with $0 \leq d \leq \frac{1}{2}$, then Bidder~2 can either win by 
    bidding $b_2^1>d$ or lose by bidding $b_2^1<d$. In the former
    case, the adversary's budget in the second auction is $B-b_2^1$, so Bidder~1's profit from the 
    second item is $\pi_1 = \frac{1}{2} - (B-b_2^1) = \frac{1}{2}-B+b_2^1$. This is 
    minimized (with value $\frac{1}{2}-B+d$) when Bidder~2 bids an amount negligibly 
    larger than $d$. In the latter case, the adversary loses the first item, so he still has budget $B$ for the 
    second auction. Thus Bidder~1 loses the second item and makes no profit on it. Bidder~1's total profit is 
    then $\frac{1}{2}-d$. Since $\frac{1}{2}-B+d$ is increasing in $d$ and $\frac{1}{2}-d$ is decreasing in $d$, her profit $\pi_1$ is 
    maximized at $d = \frac{B}{2}$. This gives a maximum guaranteed profit of $\pi_1^* = \frac{1}{2}-\frac{B}{2}$.
\end{itemize}
Putting this all together we have that
\begin{equation} \label{ex:twoitem}
    \centering
    {
    \renewcommand{\arraystretch}{1.2}
    \begin{tabular}{|c||c|c|c|c|}
        \hline
        Budget& $0 \leq B < \frac{1}{4}$ & $\frac{1}{4} \leq B < \frac{1}{2}$ & $\frac{1}{2} \leq B < 1$ & $ 1 \leq B$ \\
        \hline
        Profit $\pi_1^*$ & $1-2B$ & $\frac{3}{4}-B$ & $\frac{1}{2}-\frac{B}{2}$ & 0 \\
        \hline
    \end{tabular}
    }
\end{equation}

Before proceeding to the case of the uniform additive auction with three items, we emphasize that even the very simple case $\mathcal{A}_2$ 
illustrates many of the strategic considerations that arise in more complex sequential auctions.
To wit, in the first time period Bidder~1 faces the standard conundrum that bidding high increases her chances of winning but at 
the expense of receiving a smaller profit if she does win. More interestingly, in this adversarial setting, Bidder~1 has an additional
incentive for bidding high. If she bids high {\em and} loses then she faces a weaker adversary in the subsequent time period.
That is, by winning the first item at a high price the adversary's budget is significantly reduced in the auction for the second item. 
Counterintuitively, therefore, in adversarial sequential actions, Bidder~1 has an incentive to lose some of the items!

Interestingly, the adversary has perhaps even stronger incentives to lose than Bidder~1.
Whilst winning the first item does hurt Bidder~1, by reducing his budget, this also reduces the strength of the
adversary in the subsequent round.
Thus, the optimal outcome for adversary is that he lose the first item at a high price; this keeps the profit
of Bidder~ 1 low and increases the relative strength of the adversary in the second auction. This is in stark contrast with the 
simultaneous case, where both Bidder~1 and the adversary have an incentive to win every item.

We remark that these basic motivations and incentives play a fundamental role in sequential auctions with many items. 
That said, let's now proceed to the case $\mathcal{A}_3$ of three items.\medskip

\noindent{\tt Three Items.} Now there are three items $a_1, a_2$ and $a_3$ 
and Bidder~1 has an additive valuation function with $v(\{a_1\}) = v(\{a_2\})  = v(\{a_3\}) =\frac{1}{3}$. 
Applying a similar case analysis, her maximum guaranteed profits are then:
{\small
\begin{equation} \label{ex:threeitem}
    \centering
    {
    \renewcommand{\arraystretch}{1.5}
    \begin{tabular}{|c|c|}
        \hline
        Budget & Profit $\pi_1^*$ \\
        \hline \hline
        $0 \leq B < \frac{1}{9}$ & $1-3B$ \\
        \hline
        $\frac{1}{9} \leq B < \frac{1}{6}$ & $\frac{8}{9}-2B$ \\
        \hline
        $\frac{1}{6} \leq B < \frac{1}{3}$ & $\frac{7}{9}-\frac{4B}{3}$ \\
        \hline
        $\frac{1}{3} \leq B < \frac{5}{9}$ & $\frac{7}{12}-\frac{3B}{4}$ \\
        \hline
        $\frac{5}{9} \leq B < \frac{2}{3}$ & $\frac{4}{9}-\frac{B}{2}$ \\
        \hline
        $\frac{2}{3} \leq B < 1$ & $\frac{1}{3}-\frac{B}{3}$ \\
        \hline
        $1 \leq B$ & 0 \\
        \hline
    \end{tabular}
    }
\end{equation}}

We remark that this profit function is still piecewise linear and is so for $\mathcal{A}_m$ in general. However the complexity of the profit function grows
rapidly as the number of items increases.

\section{Tight Bounds for XOS Valuation Functions} \label{s:xos-bounds}
In this section we prove tight bounds on the risk-free profitability of Bidder~1 with a fractionally subadditive (XOS)
valuation function. In Sections~\ref{ss:xos-lb} and~\ref{ss:xos-ub} we study the sequential auction setting, and in Section~\ref{ss:simultaneous} we consider the simultaneous case. Specifically, in Section~\ref{ss:xos-lb}, we show that the agent has a strategy in the normalized auction that gives a guaranteed profit of 
$(1-\sqrt{B})^2$ when the adversary has a budget of $B$. This is equivalent to a profit of $(\sqrt{B_1}-\sqrt{B_2})^2$ in the 
original (unnormalized) auction. Then, in Section~\ref{ss:xos-ub}, we prove that 
no strategy can guarantee a profit that is greater than this by an (asymptotically zero) additive quantity.

\subsection{The XOS Lower Bound} \label{ss:xos-lb}
It is quite straightforward to obtain a lower bound on the profitability of Bidder~1 when she has an XOS valuation function. She simply chooses the additive function that maximizes her valuation under the assumption that she wins every item. Then, for each item, she bids a fixed fraction of the value of this item under this additive function. For any XOS valuation, this strategy guarantees a profit of at least $(1-\sqrt{B})^2$ against any strategy utilized 
by an adversary with a budget of $B\in (0,1)$. Consequently, we have the following.

\begin{theorem}\label{thm:xos-lb}
$\mathcal{P}(XOS,B) \geq (1-\sqrt{B})^2$.
\end{theorem}
\begin{proof}
Let $I = \{a_1,\ldots,a_m\}$ be the set of auctioned items, and $v$ be Bidder~1's XOS valuation function on $I$. 
Since $v$ is XOS, there is a set $\{\gamma_1,\gamma_2, \ldots, \gamma_\ell\}$ of (normalized) additive set functions on $I$ such that
for any $S \subseteq I$ we have $v(S) = \max_{i\in[\ell]} \gamma_i(S)$. 
Let $\gamma^*= \text{arg}\hspace{-.05cm}\max_{i\in[\ell]} \gamma_i(I)$ be an additive function that 
induces the value of $v$ on the entire set of items $I$. Thus $v(I) = \gamma^*(I)$. Moreover, by definition of $v$, we 
have that
\begin{equation}\label{eq: gamma-star}
 v(S) \geq \gamma^*(S) \hspace{1.5cm} \forall S \subseteq I.
\end{equation}
Now consider the following 
strategy for Bidder~1. In time period $t$ let Bidder~1 place a bid
of $b_1^t = \sqrt{B}\cdot \gamma^*(a_t)$ on item $a_t\in I$, for all $t\in [m]$. Let $I_1 \subseteq I$ 
be the set of items won by Bidder~1 at the end of the auction (that is, by the end of time period $m$). Similarly
let $I_2\subseteq I$ be the set of items won by Bidder~2 at the end of the auction.
Therefore $I_1\cup I_2 =I$ and, by the additivity of $\gamma^*$, we have 
\begin{equation}\label{eq: sum=1}
\gamma^*(I_1) + \gamma^*(I_2)  = 1.  
\end{equation}
It follows that during the sequential auction the adversary spent
\begin{equation}\label{eq:spending}
\sum_{t\in I_2} b_2^t 
\ \ge\  \sum_{t\in I_2} b_1^t  
\ =\  \sum_{t\in I_2} \sqrt{B}\cdot \gamma^*(a_t)  
 \ =\  \sqrt{B}\cdot \sum_{t\in I_2} \gamma^*(a_t) 
\  = \ \sqrt{B}\cdot  \gamma^*(I_2).
\end{equation}
Here the inequality arises because Bidder~2 won the items in $I_2$. The first equality follows by the definition of Bidder~1's safety strategy
and the third equality follows by the additivity of $\gamma^*$. But the adversary has a total budget of $B$.
Therefore, together this budget constraint on Bidder~2 and Inequality (\ref{eq:spending}) imply that 
$\sqrt{B}\cdot \gamma^*(I_2) \leq B$. Hence, $\gamma^*(I_2) \leq \sqrt{B}$.
From Equation (\ref{eq: sum=1}) we then derive that
\begin{equation}\label{eq: lower-I_1}
\gamma^*(I_1) \geq (1-\sqrt{B}).
\end{equation}
Now define $\pi_1$ to be the total profit obtained by Bidder~1 using this safety strategy. Then 
\begin{align}\label{eq:pi}
    \pi_1 
    &= v(I_1) - \sum_{t\in I_1} b_1^t  
         =  v(I_1) -\sum_{t\in I_1} \sqrt{B}\cdot \gamma^*(a_t) \nonumber \\ 
    &=  v(I_1) - \sqrt{B}\cdot \gamma^*(I_1) 
    \geq (1-\sqrt{B})\cdot \gamma^*(I_1).
\end{align}
Here the inequality follows from the property (\ref{eq: gamma-star}) applied to the subset $I_1$.
Finally, combining (\ref{eq: lower-I_1}) and (\ref{eq:pi}) gives
$\pi_1 \ge (1-\sqrt{B})^2$, 
as required.
\end{proof}

Next we will show this bound is tight by providing a construction where the adversary has a strategy limiting the
profitability of Bidder~$1$ to this quantity. This is surprising because the bidding strategy described above is {\em non-adaptive} -- it does not adapt to 
the history of the auction. Given the extra flexibility afforded by adaptive strategies, one would expect a priori the optimal 
risk-free strategy to be adaptive. In fact, the result that follows from this construction is doubly surprising, as it holds even if the adversary commits to his deterministic strategy in advance and Bidder~1 is allowed to randomize her strategy. Unlike the simultaneous case (which we discuss later), in this sequential setting the simple bidding strategy presented above is optimal for Bidder~$1$, and she can obtain no improvement with an adaptive or randomized strategy.

\subsection{The XOS Upper Bound} \label{ss:xos-ub}
In this section we present a matching upper bound, showing that the highest guaranteed profit of a risk-free strategy in a 
normalized two-player sequential auction with an XOS valuation is within a $\frac{1}{\sqrt{m}}$-additive factor of our lower bound. 
To do this, we present a sequential auction with an XOS valuation function 
where the game value is at most $(1-\sqrt{B})^2+ \frac{1}{\sqrt{m}}$. 
In fact, rather surprisingly, this upper bound applies even for additive valuation functions. 
Specifically, we prove that for the uniform additive auction $\mathcal{A}_m$ 
the adversary has a strategy that ensures that the profit of Bidder~1 is at most $(1-\sqrt{B})^2+ \frac{1}{\sqrt{m}}$. Consequently, 
the upper bound applies to every class of valuation functions that contains the additive functions! Together with the lower bound, 
this resolves the profitability of several well-studied classes, including the additive, submodular and gross substitutes valuation classes. 
We will see later on, in Section \ref{s:subadd-bounds}, that the situation is not as simple for subadditive valuations (that are not contained in XOS). 
We denote by $XOS_m$ the class of $XOS$ functions on $m$ items. The following 
theorem, together with the lower bound, gives our main result: $\mathcal{P}(XOS,B)$ is asymptotically equal to $(1-\sqrt{B})^2$ when $B \in (0,1)$.

\begin{theorem} \label{thm:xos-ub}
$\mathcal{P}(XOS_m,B) \leq (1-\sqrt{B})^2 + \frac{1}{\sqrt{m}}$.
\end{theorem}
\begin{proof}
We prove this result by induction on $m$. To do this, we start with a simple observation: after the first item has been sold in the 
uniform additive auction $\mathcal{A}_m$ then the sequential auction on items $\{a_2,\dots, a_m\}$ is simply
the auction $\mathcal{A}_{m-1}$ but with the additive values scaled by a multiplicative factor $\frac{m-1}{m}$; that is,
the agent now has a value $\frac{1}{m}$ for each item rather than $\frac{1}{m-1}$ as in the unscaled $\mathcal{A}_{m-1}$.
Consequently, by appropriately scaling the values {\em and} the budget of the adversary we will be able to 
analyze the auction $\mathcal{A}_m$ by studying the first round of that auction and then applying induction on the remaining
rounds.

Formally, for any positive integer $m$ let $f_m:\mathbb{R}_{\geq0}\rightarrow[0,1]$ be a function giving the highest 
guaranteed profit $f_m(x)$ of a risk-free strategy given that the adversary has a budget $B=x$.
 Clearly, for all $m$, we have that $f_m(0) = 1$ and that $f_m(x) = 0$ for any $x\ge1$. Set $f(x)= (1-\sqrt{x})^2$.
Then we want to prove by induction that
\begin{equation}\label{eq:f-bound}
f_m(x) \le f(x)+\frac{1}{\sqrt{m}} \hspace{2cm}\forall m\ge 1, \forall x\in (0,1).
\end{equation}

\noindent{\tt Base Cases:}
For the base cases, consider $m\in\{1,2,3\}$. Note that we have already studied the auctions $\mathcal{A}_1, \mathcal{A}_2$ 
and $\mathcal{A}_3$ in Section \ref{s:example}. Specifically, we found that$f_1(x) = (1-x)$, and that $f_2(x)$ is given 
by (\ref{ex:twoitem}) and $f_3(x)$ is given by (\ref{ex:threeitem}). It can be easily verified (see Figure \ref{fig:fT}) 
that each of the above functions $f_m(x)$, $m\in\{1,2,3\}$, is at 
most $f(x) + \frac{1}{\sqrt{m}}$, for any $x \in [0,1]$. Consequently, the base cases hold. \medskip

\begin{figure}[ht]
    \centering
    \begin{tikzpicture}
        \begin{axis}
        [
            width=8cm,
            height=6cm,
            title={Budget},
            title style=
            {
                at={(current axis.south)},
                anchor=north,
                outer sep=.9cm,
            },
            xmin=0,
            xmax=1.2,
            ymin=0,
            ymax=1.8,
            domain=0:1,
            samples=50,
        ]
        \addplot[thick]
            {1+x-2*sqrt(x)};
        
        \addplot[thick,dashed]
            {1+x-2*sqrt(x) + 0.707107};
        \addplot[thick,dashed,mark=*]
            coordinates {
                (0,1)
                (0.25,0.5)
                (0.5,0.25)
                (1,0)
            };
            
        \addplot[thick,dotted]
            {1+x-2*sqrt(x) + 0.577350};
        \addplot[thick,dotted,mark=*]
            coordinates {
                (0,1)
                (0.111111,0.666667)
                (0.166667,0.555556)
                (0.333333,0.333333)
                (0.555556,0.166667)
                (0.666667,0.111111)
                (1,0)
            };
        
        \legend{$f(x)$,$f(x)+\frac{1}{\sqrt{2}}$,$f_2(x)$,$f(x)+\frac{1}{\sqrt{3}}$,$f_3(x)$}
        \end{axis}
    \end{tikzpicture}
    \caption{Plot of the functions $f(x)$, $f(x)+\frac{1}{\sqrt{2}}$, $f_2(x)$, $f(x)+\frac{1}{\sqrt{3}}$ and $f_3(x)$.}
    \label{fig:fT}
\end{figure}
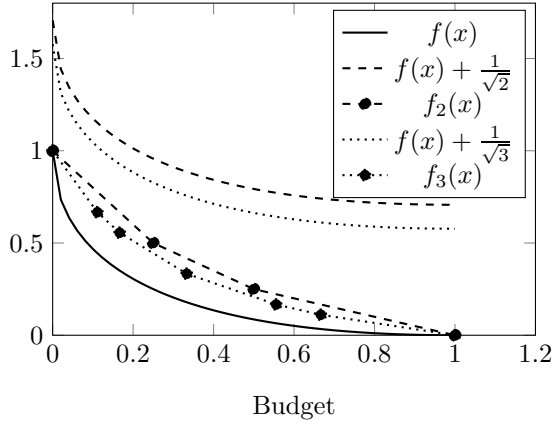

\noindent{\tt Induction Hypothesis:}
Assume that $f_k(x) \leq f(x) + \frac{1}{\sqrt{k}}$ for all $k < m$. \medskip

\noindent{\tt Induction Step:}
Given the induction hypothesis we will now prove that $f_m(x) \leq f(x) + \frac{1}{\sqrt{m}}$.
We will present a strategy for the adversary and prove that this strategy guarantees that Bidder~1 cannot make 
a profit greater than $f(x) + \frac{1}{\sqrt{m}}$ in the uniform additive auction $\mathcal{A}_m$. Specifically, we 
consider the auction for the first item $a_1$ in $\mathcal{A}_m$, and we let $b^1_2= \alpha\cdot \frac{1}{m}$ 
be the adversary's bid on this item. Since Bidder~1 has an additive value $\frac{1}{m}$ for this item, the adversary 
will never make a bid $b^1_2> \frac{1}{m}$. Thus we may assume that the adversary makes a 
bid $b^1_2= \alpha\cdot \frac{1}{m}$ for some $0\le \alpha\le 1$. We then show that for some 
particular choice of $\alpha$, even with an optimal response Bidder~1 does not make a profit 
greater than $f(x) + \frac{1}{\sqrt{m}}$. In determining her optimal response, Bidder~1 faces the 
dilemma of whether or not to outbid the adversary.
Thus we have two possibilities: \medskip

\noindent$\bullet$ {\em Bidder~1 wins item $a_1$.}\\
In this case it is easy to see that Bidder~1 will bid $b_1^1=b_2^{1+}$ (which is $b_2^1 + \epsilon$ for any negligibly 
small $\epsilon$) as any higher bid will lead to a strictly smaller profit
as this is a first-price auction.
Thus, Bidder~1  makes an immediate  profit of $\frac{1}{m} -\alpha\cdot \frac{1}{m}=\frac{1-\alpha}{m}$ on the first item. 
The rest of the sequential auction is a scaled version of $\mathcal{A}_{m-1}$.
As discussed, the additive valuations of Bidder~1 are scaled by a multiplicative factor of $\frac{m-1}{m}$.
Moreover, the budget of the adversary is also scaled. As the adversary lost the first item his budget remains $x$,
which corresponds to a budget of $B=\frac{m}{m-1}\cdot x$ in the scaled auction $\mathcal{A}_{m-1}$.
Therefore, given that the bidders play optimal strategies in the remaining rounds, the maximum profit Bidder~1 can
make is:
\begin{equation}\label{eq-g}
    g_m(x,\alpha) = \frac{1-\alpha}{m} + \frac{m-1}{m}\cdot f_{m-1} \left(\frac{mx}{m-1}\right).
\end{equation}

\noindent$\bullet$ {\em Bidder~1 loses item $a_1$.}\\
If Bidder~1 loses the first item, then Bidder~1 makes no profit on $a_1$. 
Any bid $b_1^1 < b_2^1$ will lose the item.
Since this is a first-price auction the adversary will pay $b_2^1$ if he wins regardless of the bid of Bidder~1.
Thus Bidder~1 is indifferent between any bids less than $b_2^1$. Then, after the first round we 
again have a scaled version of $\mathcal{A}_{m-1}$ where the valuations of Bidder~1 are scaled by a factor 
of $\frac{m-1}{m}$. As the adversary won the first item for a price $b_2^1=\alpha\cdot \frac{1}{m}$ his budget is now 
$x-\alpha\cdot \frac{1}{m} $,
which corresponds to a budget of $B=\frac{m}{m-1}\cdot \left(x-\frac{\alpha}{m}\right)= \frac{mx-\alpha}{m-1}$ in the scaled 
auction $\mathcal{A}_{m-1}$.
Therefore, given that the bidders play optimal strategies in the remaining rounds, the maximum profit Bidder~1 can
make is:
\begin{equation}\label{eq-h}
    h_m(x,\alpha) = \frac{m-1}{m}\cdot f_{m-1}\left(\frac{mx-\alpha}{m-1}\right).
\end{equation}

Evidently, the best response of Bidder~1 to a bid $b_2^1=\alpha\cdot\frac{1}{m}$ is given by 
the maximum of $g_m(x,\alpha)$ and $h_m(x,\alpha)$. Thus, the adversary should select $\alpha$
to minimize this maximum. Specifically,
\begin{align*}
&f_m(x) 
= \min_{0\le \alpha\le 1}\, \max \Big( g_m(x,\alpha)\, , \, h_m(x,\alpha) \Big)\\
&=\min_{0\le \alpha\le 1}\, \max \left( \frac{1-\alpha}{m} + \frac{m-1}{m} f_{m-1} \left(\frac{mx}{m-1}\right)\, , \,
\frac{m-1}{m} f_{m-1}\left(\frac{mx-\alpha}{m-1}\right)\right).
\end{align*}

Thus, our goal is to prove that there exists a bid $b_2^1=\tilde{\alpha}\cdot \frac{1}{m}$ by the adversary 
such that both $g_m(x,\tilde{\alpha})$ and $h_m(x,\tilde{\alpha})$ are at most $f(x) + \frac{1}{\sqrt{m}}$. 
This will ensure that the maximum guaranteed profit of Bidder~1 is $f_m(x)\le f(x) + \frac{1}{\sqrt{m}}$ as required.

Our proof of this fact requires examination of three cases depending upon the magnitude of the budget of
the adversary. 
In the first two cases, where the adversary has either a very low budget or a very high budget we can compute
exactly the bids that Bidder~1 will make in her unique risk-free strategy. These two cases 
do not require the induction hypothesis (nor consideration of the functions $g_m(x,\tilde{\alpha})$ and $h_m(x,\tilde{\alpha})$) but 
constitute a part of our inductive step. The third case, where the 
adversary has an intermediate budget,
is more difficult and represents one of the main technical contributions of this paper. \medskip

\noindent{\tt Low Budget Case:} $0 \leq x < \frac{1}{m^2}$.\\
When the adversary's budget $x$ is less than $\frac{1}{m^2}$ then a risk-free 
strategy for Bidder~1 is to bid $b_1^t = x$ on item $a_t$, for every $t\in[m]$ (that is, Bidder~1 bids the entire budget of the adversary on each item). 
Bidder~1 will then win all of the items for a guaranteed profit of $1-m\cdot x$. On the other hand if Bidder~1 bids an 
amount smaller than $x$ on some item, then the adversary can win this item. Even if Bidder~1 wins all the remaining items
her profit cannot exceed the total additive value of these remaining items which is $(m-1)\cdot \frac{1}{m}=1-\frac{1}{m}$. 
Because $x < \frac{1}{m^2}$, we have that $1-\frac{1}{m} < 1-m\cdot x$. As this is a first-price auction, 
Bidder~1 can never benefit by bidding strictly more than $x$ on any item.
If follows that the maximum profit the Bidder~1 can obtain is $f_m(x)=1-m\cdot x$.

It remains to show that $1-mx \leq (1-\sqrt{x})^2+\frac{1}{\sqrt{m}}$ in this low budget case where $0 \leq x < \frac{1}{m^2}$. 
We prove this statement by partitioning the interval $[0,\frac{1}{m^2})$ at the two points $\frac{1}{1.4m^2}$ and $\frac{1}{1.1m^2}$ into a 
collection of three sub-intervals $\mathcal{I} = \{[0,\frac{1}{1.4m^2}),[\frac{1}{1.4m^2},\frac{1}{1.1m^2}),[\frac{1}{1.1m^2},\frac{1}{m^2})\}$. 
We can then verify separately in each sub-interval that when $x$ falls inside this interval, we have $1-mx \leq (1-\sqrt{x})^2+\frac{1}{\sqrt{m}}$.

For any sub-interval in this collection, let $c(m)$ and $d(m)$ be the endpoints of the sub-interval. When $c(m) \leq x < d(m)$, we have 
\begin{align*}
    (1-\sqrt{x})^2+\frac{1}{\sqrt{m}} &= 1+x-2\sqrt{x}+\frac{1}{\sqrt{m}} &\\
    &\geq 1+c(m)-2\sqrt{d(m)}+\frac{1}{\sqrt{m}} &\\
    &\geq 1 \geq 1-mx.
\end{align*}
Here the first inequality arises because $x \geq c(m)$ and $x < d(m)$; the second inequality applies when $m=3$ and $[c(m),d(m)) \in \mathcal{I}$.
Now it is easy to verify that since $c(m)-2\sqrt{d(m)}+\frac{1}{\sqrt{m}} = 0$ has no real roots greater than $3$ for 
every $[c(m),d(m)) \in \mathcal{I}$, we have $1+c(m)-2\sqrt{d(m)}+\frac{1}{\sqrt{m}} > 1$ for all $m > 3$. Thus we 
have $1-mx \leq (1-\sqrt{x})^2+\frac{1}{\sqrt{m}}$ for all values of $x$ in the interval $[0,\frac{1}{m^2})$.\medskip
 
\noindent{\tt High Budget Case:} $\frac{m-1}{m} < x \leq 1$.\\ 
Suppose that the budget $x$ of the adversary is between $\frac{m-1}{m}$ and $1$. Then a risk-free strategy for 
Bidder~1 is to bid $b_1^t = \frac{x}{m}$ on item $a_t$ for every $t\in[m]$.  The guaranteed profit of this strategy is $\frac{1-x}{m}$
because Bidder~1 will win exactly one item using this strategy. 
To see this, observe that if the adversary wins the first item at the price $\frac{x}{m}$ then his scaled budget in the 
subsequent subgame is exactly $\frac{m}{m-1}\cdot (x-\frac{x}{m})=x$. 
If the adversary loses the first item then his scaled budget in the subsequent subgame is 
$\frac{m}{m-1}\cdot x> \frac{m}{m-1}\cdot \frac{m-1}{m}=1$; so the adversary will win all the remaining items
and hence Bidder~1 cannot win more than one item. But iterating this argument we see that if the adversary 
wins all of the first $m-1$ items then his (unscaled) budget for the final item is just $\frac{x}{m}$ and so Bidder~1 will win the final item
for a profit of $\frac{1-x}{m}$ as required. It is easy to see that no bidding strategy for Bidder~1 guarantees a higher profit: with 
lower bids, Bidder~2 wins each item and ends up with a higher budget in the subsequent subgames; with higher bids, Bidder~1 wins a single item for smaller profit.
    
It remains to show that $f_m(x)=\frac{1-x}{m} \leq (1-\sqrt{x})^2+\frac{1}{\sqrt{m}}$ when $\frac{m-1}{m} < x \leq 1$. We have
     \begin{align*}
        (1-\sqrt{x})^2+\frac{1}{\sqrt{m}} &= 1+x-2\sqrt{x}+\frac{1}{\sqrt{m}} \\
        &\geq 1+\frac{m-1}{m}-2+\frac{1}{\sqrt{m}} \\
        &= \frac{1}{\sqrt{m}} - \frac{1}{m} \\
        &\geq \frac{1}{m^2} \geq \frac{1-x}{m}.
    \end{align*}
Above, the first inequality holds because $\frac{m-1}{m} < x \leq 1$; the second inequality holds when $m\ge 3$; the third inequality holds 
since $x > \frac{m-1}{m}$.
Thus we have $\frac{1-x}{m} \leq (1-\sqrt{x})^2+\frac{1}{\sqrt{m}}$ when $\frac{m-1}{m} < x \leq 1$.\medskip

\noindent{\tt Intermediate Budget Case:} $\frac{1}{m^2} \leq x \leq \frac{m-1}{m}$.\\ 
Recall that by the induction hypothesis $f_{m-1}(x) \leq f(x) + \frac{1}{\sqrt{m-1}}$. Our goal now is to prove 
that $f_{m}(x) =  \min_{0\le \alpha\le 1}\, \max \big( g_m(x,\alpha)\, , \, h_m(x,\alpha) \big) \leq f(x) + \frac{1}{\sqrt{m}}$ 
when $\frac{1}{m^2} \leq x \leq \frac{m-1}{m}$.
Rather than calculate $f_{m}(x)$ exactly, our approach is to find a feasible choice $\tilde{\alpha}$ for the adversary that
ensures that both $g_m(x,\tilde{\alpha})$ and $h_m(x,\tilde{\alpha})$ are at most $f(x) + \frac{1}{\sqrt{m}}$. 
To do this, we begin by investigating the properties of the functions $g_m(x,\alpha)$ and $h_m(x,\alpha)$.
Using these properties, we find a candidate choice $\tilde{\alpha}$ which we first prove is feasible
and second prove gives the desired upper bound.

Let's start by showing that $g_m(x,\alpha)$ and $h_m(x,\alpha)$ are both monotonic functions.
To see this, observe that, for any fixed $m$, since the valuation function is additive and the 
space of available strategies for the adversary is constrained only by his budget, any strategy that is available to 
adversary with budget $\bar x < x$ is also available when his budget is $x$. Hence, the function $f_m$ is 
non-increasing in $x$. Therefore $g_m(x,\alpha)$ is non-increasing in $\alpha$ and $h_m(x,\alpha)$ is 
non-decreasing in $\alpha$, for any fixed $x$.

Now the minimum choice the adversary can make for $\alpha$ is zero.
So suppose the adversary bids $b_2^1 = \alpha\frac{1}{m}$ for the item $a_1$ with $\alpha=0$.
Then clearly $g_m(x,0) = \frac{1}{m} + \frac{m-1}{m}f_{m-1}\left(\frac{mx}{m-1}\right)$ 
and $h_m(x,0) = \frac{m-1}{m}f_{m-1}\left(\frac{mx}{m-1}\right)$.
Consequently, $g_m(x,0) \geq h_m(x,0)$.

On the other hand, consider the maximum choice the adversary can make for $\alpha$. We denote this value by $\alpha_{max}$.
We have two cases. 
\begin{itemize}
\item $x \geq \frac{1}{m}$\\
Then the adversary may set $\alpha=1$ and bid $\frac{1}{m}$ on the first item. In this case, both $g_m(x,1)$ 
and $h_m(x,1)$ are well defined, and we have $g_m(x,1) = \frac{m-1}{m}f_{m-1}\left(\frac{mx}{m-1}\right)$ 
and $h_m(x,1) = \frac{m-1}{m}f_{m-1}\left(\frac{mx-1}{m-1}\right)$. Because $f_{m-1}$ is non-increasing, 
we have that $g_m(x,1) \leq h_m(x,1)$. 
\item $x < \frac{1}{m}$\\
Now, by the budget constraint, the maximum possible value of $\alpha$ is $mx$. 
We want to show that $g_m(x,mx) \leq h_m(x,mx)$. To see this, suppose the adversary bids $x$ on the first item (corresponding to 
the choice $\alpha = mx$) and loses. 
Bidder~1 then makes a profit of $\frac{1}{m}-x$ on the first item. The adversary can subsequently play the following strategy: bid $x$ on 
every item until he wins an item. Of course, Bidder~1 will then win the remaining items for free after the adversary wins one item, 
because the budget of the adversary has then fallen to $0$. Now, if Bidder~1's risk-free strategy is to win all of the items at price at least $x$, her (absolute) 
profit in this subgame is at most $\left(\frac{m-1}{m}-(m-1)x\right)$. If instead the adversary wins the $k^{\text{th}}$ item, 
where $1 \leq k \leq m-1$, then Bidder~1's profit is at most $\frac{m-2}{m}-(k-1)x$, which is maximized at $k=1$ with 
value $\frac{m-2}{m}$. In both cases, Bidder~1's total 
profit on all $m$ items is either at most $1-mx$ or at most $\frac{m-1}{m}-x$. But these are both at most the profit Bidder~1 gets (namely, $\frac{m-1}{m}$) if 
she gives up the first item at price $x$ and wins the remaining $m-1$ items for free. Thus $g_m(x,mx)\le h_m(x,mx)$.
\end{itemize}

Set $\alpha_{max}$ to be the highest possible value of $\alpha$ for which both $g_m(x,\alpha)$ and $h_m(x,\alpha)$ are well-defined for all $x$.
Therefore $\alpha_{max} = \min(1,mx)$. We have shown that $g_m(x,0) \geq h_m(x,0)$ 
and $g_m(x,\alpha_{max}) \leq h_m(x,\alpha_{max})$. Then, because $g_m(x,\alpha)$ is non-increasing in $\alpha$ and $h_m(x,\alpha)$ is 
non-decreasing in $\alpha$ for fixed $x$, our upper bound of $\max(g_m(x,\alpha),h_m(x,\alpha))$ is minimized at any 
bid $\bar\alpha\cdot \frac{1}{m}$ such 
that $0 \leq \bar\alpha \leq \alpha_{max}$ and $g_m(x,\bar\alpha)=h_m(x,\bar\alpha)$. This is also precisely equal to 
a risk-free bid $\alpha^*\cdot \frac{1}{m}$ placed by Bidder~1 on the first item, since from her perspective, if the adversary plays a 
best response then she gets the \textit{minimum} of $g_m(x,\alpha^*)$ and $h_m(x,\alpha^*)$, and this minimum is 
maximized when they are equal.

We now use the above observations to establish an upper bound on the highest guaranteed profit of a 
risk-free strategy. For an appropriately chosen bid $\tilde{\alpha}\frac{1}{m}$, we prove that
both $g_m(x,\tilde{\alpha})$ and $h_m(x,\tilde{\alpha})$ are well-defined for all $x \in [\frac{1}{m^2}, \frac{m-1}{m}]$. We then 
prove that both these values are at most $f(x) + \frac{1}{\sqrt{m}}$. 
The facts rely on the four technical claims below.

The first two claims show that $\tilde{\alpha} = 1 - 2m(1-\sqrt{x}) + 2\sqrt{m(m-1)}(1-\sqrt{x})$.
is a feasible choice for $\tilde{\alpha}$: specifically, $0\le \tilde{\alpha} \le \alpha_{max}$. Let $\tilde{\alpha} = 1 - 2m(1-\sqrt{x}) + 2\sqrt{m(m-1)}(1-\sqrt{x})$. We show the following.

\begin{claim} \label{clm:price1}
For any $x \in [\frac{1}{m^2},\frac{m-1}{m}]$, $0 \leq \tilde{\alpha}$.
\end{claim}

\begin{proof}
To prove that $\tilde{\alpha}$ is non-negative, for $x \in [\frac{1}{m^2},\frac{m-1}{m}]$, we require that
\begin{equation*}
   1-2m(1-\sqrt{x}) + 2\sqrt{m(m-1)}(1-\sqrt{x}) \geq 0 
\end{equation*}    
Equivalently, we want to show that
   \begin{equation*}
2\cdot(1-\sqrt{x})\cdot\left(m-\sqrt{m(m-1)}\right) \leq 1.
\end{equation*}
Since $x \geq \frac{1}{m^2}$, we have
\begin{alignat*}{2}
    2\cdot (1-\sqrt{x})\cdot\left(m-\sqrt{m(m-1)}\right) &\leq 2\cdot \left(1-\sqrt{\frac{1}{m^2}}\right)\cdot \left(m-\sqrt{m(m-1)}\right) &\\
     &= 2\cdot \left(\frac{m-1}{m}\right)\cdot \left(m-\sqrt{m(m-1)}\right) &\\
    &= 2\cdot\left((m-1) - (m-1)\sqrt{\frac{m-1}{m}}\right) &\\
    &= 2\cdot (m-1)\cdot \left(1-\sqrt{\frac{m-1}{m}}\right).
\end{alignat*}
Thus we must show that $2(m-1)\left(1-\sqrt{\frac{m-1}{m}}\right) \leq 1$. Equivalently we require that
\begin{alignat*}{2}
\frac{m-1}{m} &\ge \left(1-\frac{1}{2\cdot(m-1)}\right)^2 &\\
& = 1-\frac{1}{m-1}+ \frac{1}{4\cdot(m-1)^2} &\\
& = \frac{m-2}{m-1}+ \frac{1}{4\cdot(m-1)^2}
\end{alignat*}  
To prove this, observe that $\frac{m-1}{m}- \frac{m-2}{m-1}=\frac{1}{m\cdot(m-1)}$, for all $m\ge 2$. Therefore
\begin{alignat*}{2}
\frac{m-1}{m} &= \frac{m-2}{m-1}+\frac{1}{m\cdot(m-1)} &\\
&\ge \frac{m-2}{m-1}+\frac{1}{4\cdot(m-1)^2}
\end{alignat*}
where the inequality holds for all $m \geq 2$. This proves that $\tilde{\alpha} \geq 0$. 
\end{proof}

\begin{claim} \label{clm:price2}
For any $x \in [\frac{1}{m^2},\frac{m-1}{m}]$, $\tilde{\alpha} \leq \alpha_{max}$.
\end{claim}

\begin{proof}
We partition the proof into two cases.\medskip

\noindent (i) Assume that $x \geq \frac{1}{m}$. It follows that $\alpha_{max} = 1$ and so we must show
   \begin{equation*}
    1-2m(1-\sqrt{x}) + 2\sqrt{m(m-1)}(1-\sqrt{x})\leq 1 
\end{equation*}  
Equivalently, we require  
       \begin{equation*}
2\cdot (1-\sqrt{x})\cdot \left(m-\sqrt{m(m-1)}\right) \geq 0
\end{equation*}
But this inequality holds because, by assumption, we have $x\in [\frac{1}{m}, \frac{m-1}{m}]$.  In particular, for $m\geq 2$,
we have both $(1-\sqrt{x}) > 0$ and $(2m-2\sqrt{m(m-1)}) > 0$.\medskip

\noindent (ii) Assume that  $x < \frac{1}{m}$.
 It now follows that $\alpha_{max} = m\cdot x$ and so we must show
   \begin{equation*}
    1-2m(1-\sqrt{x}) + 2\sqrt{m(m-1)}(1-\sqrt{x}) \leq m\cdot x 
    \end{equation*}
 Equivalently, we require     
      \begin{equation*}
2\cdot (1-\sqrt{x})\cdot\left(m-\sqrt{m(m-1)}\right) + m\cdot x - 1 \geq 0.
\end{equation*} 
To see this holds, observe that
\begin{eqnarray*}
\lefteqn{2\cdot (1-\sqrt{x})\cdot \left(m-\sqrt{m(m-1)}\right) + m\cdot x - 1}\\
&=& mx -2\sqrt{x}\cdot \left(m-\sqrt{m(m-1)}\right)+ \left(2\cdot(m-\sqrt{m(m-1)})-1\right) \\
&=& mx -2\sqrt{mx}\cdot \left(\sqrt{m}-\sqrt{m-1}\right)+ \left(2m-1- 2\cdot \sqrt{m(m-1)}\right) \\
&=& \left(\sqrt{mx}-\left(\sqrt{m}-\sqrt{m-1}\right)\right)^2 \\
&\geq& 0
\end{eqnarray*}
The claim follows.
\end{proof}

So we have a feasible choice for $\tilde{\alpha}$.
To complete the proof that $f_m(x) \leq f(x) + \frac{1}{\sqrt{m}}$, we now show that both 
$g_m(x,\tilde{\alpha})$ and $h_m(x,\tilde{\alpha})$
 are at most $f(x) + \frac{1}{\sqrt{m}}$. 

\begin{claim} \label{clm:ub1}
$g_m(x,\tilde{\alpha}) \leq f(x) + \frac{1}{\sqrt{m}}$.
\end{claim}

\begin{proof}
By our induction hypothesis,
\begin{align*}
    g_m(x,\tilde{\alpha}) &= \frac{1-\tilde{\alpha}}{m} + \frac{m-1}{m}\cdot f_{m-1}\left(\frac{mx}{m-1}\right) &\\
    &\leq \frac{1-\tilde{\alpha}}{m} + \frac{m-1}{m}\cdot \left(f\left(\frac{mx}{m-1}\right)+\frac{1}{\sqrt{m-1}}\right)
\end{align*}
Thus it suffices to show
\begin{align*}
    \frac{1-\tilde{\alpha}}{m} + \frac{m-1}{m}\cdot \left(f\left(\frac{mx}{m-1}\right)+\frac{1}{\sqrt{m-1}}\right) &\leq f(x) + \frac{1}{\sqrt{m}}
\end{align*}
Equivalently, we require
\begin{align*}
    f(x) + \frac{1}{\sqrt{m}} - \frac{1-\tilde{\alpha}}{m} - \frac{m-1}{m}\cdot \left(f\left(\frac{mx}{m-1}\right)+\frac{1}{\sqrt{m-1}}\right) &\geq 0
\end{align*}
To prove this, observe that
\begin{align*}
   f(x) &+ \frac{1}{\sqrt{m}} - \frac{1-\tilde{\alpha}}{m} - \frac{m-1}{m}\left(f\left(\frac{mx}{m-1}\right)+\frac{1}{\sqrt{m-1}}\right) \\
    &= \Big(1 + x - 2\sqrt{x}\Big) - \frac{1-\tilde{\alpha}}{m}  - \frac{m-1}{m} \left( 1 + \frac{mx}{m-1} -2\sqrt{\frac{mx}{m-1}} \right) \\
    &\qquad\qquad + \frac{1}{\sqrt{m}} - \frac{\sqrt{m-1}}{m} \\
    &= 1 + x - 2\sqrt{x} - \frac{1}{m} + \frac{\tilde{\alpha}}{m} - 1 + \frac{1}{m} - x + 2\sqrt{\frac{(m-1)x}{m}} \\
   &\qquad\qquad + \frac{1}{\sqrt{m}} - \frac{\sqrt{m-1}}{m} \\
    &= -2\sqrt{x} + \frac{\tilde{\alpha}}{m} + 2\sqrt{\frac{(m-1)x}{m}} 
 + \frac{1}{\sqrt{m}}\left(1-\sqrt{\frac{m-1}{m}}\right) 
\end{align*} 

By the definition of $\tilde{\alpha}$ we then have that
\begin{align*}
 f(x) &+ \frac{1}{\sqrt{m}} - \frac{1-\tilde{\alpha}}{m} - \frac{m-1}{m}\left(f\left(\frac{mx}{m-1}\right)+\frac{1}{\sqrt{m-1}}\right)\\
    &=-2\sqrt{x} + \left(\frac{1}{m} -2+2\sqrt{x} + \frac{2\sqrt{m-1}}{\sqrt{m}}\left(1-\sqrt{x}\right)\right) \\
    &\qquad\qquad + 2\sqrt{\frac{(m-1)x}{m}}+ \frac{1}{\sqrt{m}}\left(1-\sqrt{\frac{m-1}{m}}\right) \\
    &= \frac{1}{m} - 2 + 2\sqrt{\frac{m-1}{m}}(1-\sqrt{x}) + 2\sqrt{\frac{m-1}{m}}\sqrt{x}
  + \frac{1}{\sqrt{m}}\left(1-\sqrt{\frac{m-1}{m}}\right) \\
    &= \frac{1}{m} - 2\cdot\left(1-\sqrt{\frac{m-1}{m}}\right)  + \frac{1}{\sqrt{m}}\left(1-\sqrt{\frac{m-1}{m}}\right) \\
    &= \frac{1}{m} + \left(\frac{1}{\sqrt{m}}-2\right)\cdot \left(1-\sqrt{\frac{m-1}{m}}\right)
\end{align*}
Now set $q(m) = \frac{1}{m} + \left(\frac{1}{\sqrt{m}}-2\right)\cdot \left(1-\sqrt{\frac{m-1}{m}}\right)$. 
Clearly, to show that $q(m)$ is non-negative, it suffices to show that $(\sqrt{m} + \sqrt{m-1})  \cdot m\cdot q(m)$ is non-negative. 
To do this, note that
\begin{align*}
(\sqrt{m} &+ \sqrt{m-1})  \cdot m\cdot q(m) \\
    &= (\sqrt{m} + \sqrt{m-1})  \cdot\bigg(1 +(1-2\sqrt{m})\cdot(\sqrt{m} - \sqrt{m-1})\bigg)\\
    &= (\sqrt{m} + \sqrt{m-1}) - (2\sqrt{m}-1)\cdot(\sqrt{m} - \sqrt{m-1})\cdot (\sqrt{m} + \sqrt{m-1})\\
    &= (\sqrt{m} + \sqrt{m-1})-(2\sqrt{m}-1)\cdot 1\\
    &= 1 + \sqrt{m-1} - \sqrt{m}\\
    &\ge 0
\end{align*}
Here the final inequality holds for $m \geq 2$. Thus $g_m(x,\tilde{\alpha}) \leq f(x) + \frac{1}{\sqrt{m}}$.
\end{proof}

\begin{claim} \label{clm:ub2}
$h_m(x,\tilde{\alpha}) \leq f(x) + \frac{1}{\sqrt{m}}$.
\end{claim}

\begin{proof}
We now prove  $h_m(x,\tilde{\alpha}) \leq f(x) + \frac{1}{\sqrt{m}}$. By our induction hypothesis,
\begin{eqnarray*}
    h_m(x,\tilde{\alpha}) &=& \frac{m-1}{m}\cdot f_{m-1}\left(\frac{mx-\tilde{\alpha}}{m-1}\right) \\
    &\leq& \frac{m-1}{m}\cdot \left(f\left(\frac{mx-\tilde{\alpha}}{m-1}\right)+\frac{1}{\sqrt{m-1}}\right)
\end{eqnarray*}
Hence, we want to show that
\begin{equation*}
    \frac{m-1}{m}\cdot \left(f\left(\frac{mx-\tilde{\alpha}}{m-1}\right)+\frac{1}{\sqrt{m-1}}\right) \leq f(x) + \frac{1}{\sqrt{m}}
\end{equation*}
Equivalently, we require
\begin{align*}
    f(x) - \frac{m-1}{m}\cdot f\left(\frac{mx-\tilde{\alpha}}{m-1}\right) + \frac{1}{\sqrt{m}} - \frac{m-1}{m}\left(\frac{1}{\sqrt{m-1}}\right) &\geq 0
\end{align*}
To begin, let's show that the first two terms are equal; that is, $f(x) - \frac{m-1}{m}\cdot f\left(\frac{mx-\tilde{\alpha}}{m-1}\right)=0$.
\begin{alignat*}{2}
f(x) &- \frac{m-1}{m}\cdot f\left(\frac{mx-\tilde{\alpha}}{m-1}\right) = \left(1+x-2\sqrt{x}\right) - \frac{m-1}{m}\cdot f\left(\frac{mx-\tilde{\alpha}}{m-1}\right) &\\
&= \left(1+x-2\sqrt{x}\right) -\frac{m-1}{m}\left(1+\frac{mx-\tilde{\alpha}}{m-1}-2\sqrt{\frac{mx-\tilde{\alpha}}{m-1}}\right) &\\
&= \left(1+x-2\sqrt{x}\right) -\left(1-\frac{1}{m}+x - \frac{\tilde{\alpha}}{m}-2\cdot\frac{m-1}{m}\cdot \sqrt{\frac{mx-\tilde{\alpha}}{m-1}}\right) &\\
&= -2\sqrt{x} + \frac{1}{m}+ \frac{\tilde{\alpha}}{m}+2\cdot\frac{m-1}{m}\cdot \sqrt{\frac{mx-\tilde{\alpha}}{m-1}}
\end{alignat*}
To prove the RHS is indeed $0$ we must show that
\begin{equation*}
   2\sqrt{x} - \frac{1+\tilde{\alpha}}{m}= 2\cdot\frac{m-1}{m}\cdot \sqrt{\frac{mx-\tilde{\alpha}}{m-1}}
\end{equation*}
First observe that both sides are nonnegative, because $\sqrt{x} > \frac{1}{m}$ and $\tilde{\alpha} < \alpha_{max}$. Next
multiply each side by $m$ and take the square. This leads to 
\begin{alignat*}{2}
\lefteqn{\left(2m\cdot \sqrt{x} - (1+\tilde{\alpha})\right)^2 - 4\cdot(m-1)\cdot (mx-\tilde{\alpha})} \\
&= \left((2m\cdot \sqrt{x} -1)-\tilde{\alpha})\right)^2 - 4\cdot(m-1)\cdot (mx-\tilde{\alpha}) \\
&= \left(\tilde{\alpha}^2 - (4m\cdot \sqrt{x} -2)\tilde{\alpha} + (4m^2x-4m\sqrt{x}+1) \right) \\
&\qquad + \left(4(m-1)\cdot \tilde{\alpha}-4(m-1)mx\right) \\
&= \tilde{\alpha}^2 + (2-4m\cdot \sqrt{x} +4(m-1))\cdot \tilde{\alpha} + (4m^2x-4m\sqrt{x}+1-4(m-1)mx) \\
&= \left(\tilde{\alpha}^2 - (4m\cdot \sqrt{x} -2)\tilde{\alpha} + (4m^2x-4m\sqrt{x}+1) \right) \\
&\qquad + \left(4(m-1)\cdot \tilde{\alpha}-4(m-1)mx\right) \\
&= \tilde{\alpha}^2 + (4(m-1)-4m\cdot \sqrt{x} +2)\cdot \tilde{\alpha} + (4m^2x-4m\sqrt{x}+1-4(m-1)mx) \\
&= \tilde{\alpha}^2 + (4m-4m\cdot \sqrt{x} -2)\cdot \tilde{\alpha} + (4mx-4m\sqrt{x}+1) \\
&=0
\end{alignat*}
The final equality does follow as $z = \tilde{\alpha}$ is indeed a solution to the quadratic equation 
$z^2 + (4m-4m\sqrt{x}-2)z + (4mx-4m\sqrt{x}+1) = 0$. 
Putting this all together then gives
\begin{alignat*}{2}
    f(x) - \frac{m-1}{m}\cdot f\left(\frac{mx-\tilde{\alpha}}{m-1}\right) &+ \frac{1}{\sqrt{m}} - \frac{m-1}{m}\left(\frac{1}{\sqrt{m-1}}\right) \\
    &= 0 + \frac{1}{\sqrt{m}} - \frac{m-1}{m}\left(\frac{1}{\sqrt{m-1}}\right) \\
    &= \frac{1}{\sqrt{m}}-\frac{\sqrt{m-1}}{m} \\
    &= \frac{1}{\sqrt{m}}\left(1-\sqrt{\frac{m-1}{m}}\right) \\
    &\ge 0
\end{alignat*}
The claim follows.
\end{proof}

So we have the following upper bounds on $g_m$ and $h_m$: $g_m(x,\tilde{\alpha}) \leq f(x) + \frac{1}{\sqrt{m}}$ and $h_m(x,\tilde{\alpha}) \leq f(x) + \frac{1}{\sqrt{m}}$. 
Since we also have $f_m(x) \leq \max(g_m(x,\tilde{\alpha}),h_m(x,\tilde{\alpha}))$, we have $f_m(x) \leq f(x) + \frac{1}{\sqrt{m}}$ 
when $\frac{1}{m^2}\leq x\leq \frac{m-1}{m}$. With this third case (intermediate budget) completed so is the proof of Theorem \ref{thm:xos-ub}.
\end{proof}

\subsection{Risk-Free Bidding in Simultaneous Auctions} \label{ss:simultaneous}

In this section we consider risk-free bidding in a simultaneous auction.
For a 
budgeted adversary in a simultaneous auction,
the analogue of budget-constrained bidding is that the {\em sum} of the adversary's bids on the items is at most the budget $B$. 
Intuitively, a budgeted adversary is weaker in a simultaneous auction than in a sequential auction. This is because,
in a sequential auction, an adversary has the option to ``overbid" on an item but suffers no consequence {\em if he loses the item}.
The issue then is whether or not the resultant broader range of strategies available to an adversary in a sequential auction makes 
it provably more powerful than the corresponding adversary in a simultaneous auction. 
We show this in the following theorems. We begin by analyzing the second-price case.
\begin{theorem} \label{thm:simul-second}
    The two-player simultaneous second-price auction with a normalized XOS valuation function and an 
    adversary with normalized budget $B \in (0,1)$ has a risk-free strategy for Bidder~1 that 
    guarantees a profit of at least $(1-B)$.
\end{theorem}
\begin{proof}
We prove this theorem by using the following strategy for Bidder~1. Bidder~1 bids truthfully according to the additive function $\gamma^*$ defined in 
Section \ref{ss:xos-lb} -- that is, for each item $a_j \in I$, she bids $\gamma^*_j =\gamma^*(\{a_j\})$. We show that, if Bidder~1 plays according to this strategy, then 
for any feasible strategy of the adversary, Bidder~1 makes a profit of at least $(1-B)$. In particular, we consider the adversary's best response to this strategy.

Suppose the adversary's best response is to make a sequence $b_1,\ldots,b_m$ of bids on the respective items. Clearly, in a best response the adversary 
will not bid more than $\gamma^*(\{a_j\})^+$. Let $I_1\subseteq I$ and $I_2\subseteq I$ be the set of items allocated to Bidder~1 and Bidder~2 respectively. 
Then Bidder~1's profit is given by
\begin{align*}
    \pi_1 &= v(I_1) - \sum_{j:a_j\in I_1}b_j \geq \sum_{j:a_j\in I_1}(\gamma^*_j - b_j) & \\
    &= \sum_{j:a_j\in I_1}(\gamma^*_j - b_j) + \sum_{j:a_j\in I_2}(\gamma^*_j - b_j) &\\
    &= \sum_{j:a_j\in I}\gamma^*_j - \sum_{j:a_j\in I}b_j & \\
    &= 1 - \sum_{j:a_j\in I}b_j \geq 1 - B
\end{align*}
Here the first inequality follows by definition of an XOS function; the second equality arises because the adversary bids $b_j=\gamma^*_j$ on each item $j \in I_2$ that 
he wins; the fourth equality follows by definition of $\gamma^*$ and the fact the auction is normalized; the second inequality 
follows from the budget constraint. 
\end{proof}
Observe that $(1-\sqrt{B})^2< 1-B$, for all $B \in (0,1)$. Ergo, the risk-free profitability of Bidder~$1$ is strictly greater in a second-price simultaneous auction than in a
second-price sequential auction. Conversely the adversary is strictly weaker in the  second-price simultaneous auction.

Next let's consider the case of {\em first-price} simultaneous auctions. Note that the proof of Theorem~\ref{thm:simul-second} was via the use of a pure
strategy for Bidder~$1$. For first-price simultaneous auctions it is not possible to rely on a pure strategy to beat the profit bound of $(1-\sqrt{B})^2$;
to do so, the bidder must use a randomized strategy. To verify this, the following 
claim shows that in the {\em uniform additive} simultaneous auction, no deterministic strategy for Bidder~$1$ can guarantee a profit that is asymptotically 
greater than $(1-\sqrt{B})^2$.

\begin{claim}
For any pure strategy of Bidder~1, there exists a strategy for the adversary that (asymptotically) restricts Bidder~1's profit to $(1-\sqrt{B})^2$.
\end{claim}
\begin{proof}
    Let $b_1^1,\ldots,b_1^m$ be Bidder~1's bids on the $m$ items. We may assume without loss of generality that $b_1^i \leq b_1^j$ whenever $i < j$. 
    Bidder~2's strategy is to win $k^*$ items, where $k^* = \max\{k : \sum_{i=1}^kb_1^i < B\}$. Let $p^*$ be the price of the lowest-indexed item that 
    Bidder~1 wins, i.e., $p^* = b_1^{k^*+1}$. Since we maximize over all possible $k$, we know that the adversary cannot afford to win the entire 
    set $\{a_1,\ldots,a_{k^*+1}\}$. Let $P$ be the total price paid by the adversary. This implies that $P$ is more than $B-p^*$, otherwise the adversary 
    could have won another item. So we have $P > B-p^*$.
    
    On the other hand, the total price paid by the adversary is at most
    \begin{equation*}
        P = b_1^1 + \ldots + b_1^{k^*} \leq k^*\cdot b_1^{k^*+1} = k^*p^*.
    \end{equation*}
    Combining the inequalities, we have $k^* > \frac{B-p^*}{p^*}$. So the number of items that the adversary wins is at least $\frac{B-p^*}{p^*}$; thus, Bidder~1 wins at 
    most $m-\frac{B-p^*}{p^*}= m-\frac{B}{p^*}+1$ items. Since the items are ordered by Bidder~1's bids, her price for each of these items is at least $p^*$. 
    Consequently, because her valuation function is uniform additive, her profit is at most
    \begin{align*}
        \pi_1 \leq (m-\frac{B}{p^*}+1)\cdot (\frac{1}{m}-p^*) 
    \end{align*}
    It is easy to verify that this is maximized when $p^* = \sqrt{\frac{B}{m(m+1)}}$, and that the maximum value is $(1-\sqrt{B})^2 + O(\frac{1}{\sqrt{m}})$.
\end{proof}

Interestingly, in the other direction, it is also true that no deterministic strategy for the adversary can guarantee that Bidder~1 makes a profit that is less than $(1-B)$.
\begin{claim}
For any pure strategy of the adversary, there exists a strategy for Bidder~1 that guarantees a profit of $(1-B)$ for any XOS bidder.
\end{claim}
\begin{proof}
By the bidding constraint, the sum of the adversary's bids is at most $B$. Let $b_2^1,\ldots,b_2^m$ be the bids 
made by the adversary on the items. Bidder~1 simply bids $b_2^{i+}$ on each item $a_i$ as long as $b_2^i$ is less than $\gamma^*_i$ and wins 
this set of items. Bidder~1's profit is then
\begin{eqnarray*}
v_1(I_1)-\sum_{i:b_2^i< \gamma^*_i} b_2^i &\ge& \sum_{i:b_2^i< \gamma^*_i} (\gamma^*_i - b_2^i)\\
&\ge& \sum_{i:b_2^i< \gamma^*_i} (\gamma^*_i - b_2^i) + \sum_{i:b_2^i\ge \gamma^*_i} (\gamma^*_i - b_2^i) \\
&=& \sum_{i:a_i\in I} (\gamma^*_i - b_2^i)\\
&=& 1 - \sum_{i:a_i\in I} b_2^i\\
&\ge& 1 - B
\end{eqnarray*}
Here the first inequality arises as Bidder~$1$ has an XOS valuation; the second equality follows by definition of $\gamma^*$ and the fact the auction is normalized; 
the last inequality 
follows from the budget constraint.
\end{proof}
Due to this asymmetry in pure strategies, simultaneous first-price auctions against an adversary have no equilibrium in deterministic strategies, 
and we must introduce randomization to improve the lower bound. 
In fact, there is a randomized strategy for Bidder~1 that guarantees an expected profit of at least $\frac12 (1-B)^2$ when her valuation function is XOS. 
The function $\frac{1}{2}(1-B)^2$ is greater than $(1-\sqrt{B})^2$ for $B > 3-2\sqrt{2}$, which is approximately $0.17$, so the upper bound from the first-price sequential auction case does not 
apply to first-price simultaneous auctions.

\begin{theorem} \label{thm:simul-first}
    The two-player simultaneous first-price auction with a normalized XOS valuation function and an 
    adversary with normalized budget $B \in (0,1)$ has a (randomized) risk-free strategy for Bidder~1 that 
    guarantees a profit of at least $\frac{(1-B)^2}{2}$ in expectation.
\end{theorem}

\begin{proof}
Bidder~$1$ selects $m$ independent random variables $X_i$, each drawn from the uniform distribution $U(0,1)$, and bids $X_i\cdot\gamma^*_i $ 
on the item $a_i$. Our goal is to show that no strategy of the adversary can prevent Bidder~1 from making a profit of $\frac{1}{2}(1-B)^2$ in expectation. For the adversary, 
we may limit our attention to bids that are at most $\gamma^*_i $. So we may parameterize the bids of the adversary by a vector of 
ratios $\mbs{b} = (b_1,\ldots,b_m) \in (0,1)^m$ such that 
$\sum_{i=1}^m b_i\cdot \gamma^*(i) \leq B$. Let $S = \{ a_i \in I | X_i > b_i \}$ be the set of items that Bidder~1 wins and 
let $\pi(\mbs{b})$ be the random variable representing Bidder~1's utility when the adversary bids 
$\mbs{b}$. We then have the following.
\begin{align*}
    \mathbb{E}[\pi(\mbs{b})] &= \mathbb{E}_{X_1 \ldots X_m} \left[v(S) - \sum_{i:a_i\in S} \gamma^*_i \cdot X_i\right] &\\
    &\geq \mathbb{E}_{X_1 \ldots X_m} \left[\sum_{i:a_i\in S} \gamma^*_i \cdot (1-X_i)\right] &\\
    &= \mathbb{E}_{X_1 \ldots X_m} \left[\sum_{i:a_i\in I} \gamma^*_i\cdot (1-X_i)\cdot\mathbbm{1}_{[b_i<X_i]}\right] &\\
        &= \sum_{i:a_i\in I} \gamma^*_i \cdot \mathbb{E}_{X_i} \left[ (1-X_i)\cdot\mathbbm{1}_{[b_i<X_i]}\right] &\\
    &= \sum_{i:a_i\in I} \gamma^*_i \cdot \frac12 \cdot (1-b_i)^2 
\end{align*}
Again, here the first inequality follows from the definition of $\gamma^*$; the second equality is due to linearity of expectation;
the final equality holds because $X_i$ is uniformly distributed.

The adversary of course seeks to find a strategy to minimize $\mathbb{E}[\pi(\mbs{b})]$  
for any fixed valuation function $v$. For any valuation function $v$, we denote this minimum value by $\pi^*(v)$. The inequalities above 
imply that $\pi^*(v) \geq \pi^*(\gamma^*) = \min_{\mbs{b}} \frac{1}{2}\sum_{i:a_i \in I}  \gamma^*_i\cdot(1-b_i)^2$. For the following 
analysis, we may assume without loss of generality that we only consider items $a_i$ where $\gamma^*_i > 0$ (Bidder~1 will lose the 
remaining items at a bid of 0). 
Now, $\pi^*(v)$ is lower-bounded by the 
optimal value of the following quadratic program. 
\begin{eqnarray*}
(\text{\tt Adversarial QP})\hspace{1cm} \min\ \frac{1}{2}\cdot\sum_{i:a_i \in I}  \gamma^*_i&\cdot& (1-b_i)^2\\
\text{s. t.}\hspace{1cm} 
 b_i  &\leq& 1 \hspace{1cm} \forall i\in[m] \\
 -b_i &\leq& 0  \hspace{1cm} \forall i\in[m] \\
 \sum_{i=1}^m b_i\cdot \gamma^*(i) &\leq& B
\end{eqnarray*}
The Lagrangian of this problem is 
\begin{align*}
    \mathcal{L}(\Vec{b}, \Vec{\lambda}) &= \frac{1}{2}\sum_{i \in I}  \gamma^*_i \cdot (1-b_i)^2 + \lambda_{2m+1} \cdot \left(\sum_{i=1}^m b_i\cdot \gamma^*_i - B \right) \\
    &+ \sum_{i=1}^m \lambda_i \cdot (b_i - 1) - \sum_{i=1}^m \lambda_{m+i} \cdot b_i
\end{align*}

This is differentiable w.r.t. each of the $b_i$, so we can compute these partial derivatives 
to be
\begin{equation*}
\frac{\partial \mathcal{L}}{\partial b_i} = \gamma^*_i\cdot (b_i -1) + \lambda_{2m+1}\cdot \gamma^*_i + \lambda_i - \lambda_{m+i}.
\end{equation*}

The dual objective function $ g(\Vec{\lambda}) = \inf_{\Vec{b}} \mathcal{L}(\Vec{b}, \Vec{\lambda})$ can be computed by 
setting $\frac{\partial \mathcal{L}}{\partial b_i} = 0$. So:
\begin{align*}
g(\Vec{\lambda}) = &\frac{1}{2} \sum_{i=1}^{m} \gamma^*_i \left( \lambda_{2m+1} 
+ \left( \frac{\lambda_i - \lambda_{m+i}}{\gamma^*_i} \right)\right)^2 & \\ 
&+ \lambda_{2m+1}\left(\sum_{i=1}^{m}  \left(1- \lambda_{2m+1} - \left( \frac{\lambda_i 
- \lambda_{m+i}}{\gamma^*_i} \right)\right)\gamma^*_i - B \right) & \\
&- \sum_{i=1}^{m} \lambda_i \left( \lambda_{2m+1} + \left( \frac{\lambda_i - \lambda_{m+i}}{\gamma^*_i} \right) \right) & \\
&- \sum_{i=1}^{m} \lambda_{m+i} \left(1 - \lambda_{2m+1} - \left( \frac{\lambda_i - \lambda_{m+i}}{\gamma^*_i} \right) \right)
\end{align*}

The constraints on the dual are simply $\vec{\lambda} \geq 0$. So 
setting $\lambda_i = 0 \; \forall \; i \in \{1 \hdots 2m\}; \lambda_{2m+1} = (1-B)$ is 
feasible. Let $\vec{\lambda'}$ denote this vector. The dual objective for this feasible input is:
\begin{align*}
g(\vec{\lambda'}) &= \frac{1}{2} \sum_{i=1}^{m} \gamma^*_i \left( \lambda_{2m+1} \right)^2 &\\
&\qquad\qquad+ \lambda_{2m+1}\left(\sum_{i=1}^{m}  \left(1- \lambda_{2m+1} \right)\gamma^*_i - B \right) & \text{($\lambda_i = 0 $)}  \\
&= \frac{1}{2} \left( \lambda_{2m+1} \right)^2 + \lambda_{2m+1}\left(\left(1-\lambda_{2m+1}\right)-B\right) 
& \text{(as \ $\sum_i \gamma^*_i= 1$)} \\
&= \frac{1}{2} \left(1-B\right)^2 & \text{(substituting for $\lambda_{2m+1}$)}
\end{align*}
This dual solution lower bounds the primal minimization program, and we have
\begin{equation*}
\min_{\mbs{b}} \mathbb{E}[\pi(\mbs{b})] \geq \frac{1}{2} \left(1-B\right)^2
\end{equation*}
as desired.

\end{proof}

Finally, we show that an analogue of Theorem~\ref{thm:simul-second} does {\em not} hold for simultaneous first-price auctions. Specifically, we prove that there exists a {\em randomized} strategy for the adversary that gives an upper bound on the profitability that is strictly smaller than $(1-B)$, showing that for simultaneous auctions, the adversary's power is greater in the first-price case than the second-price case.

\begin{theorem}
In first-price simultaneous auctions with XOS valuations, the adversary has a (randomized) strategy that restricts the risk-free profit of Bidder~1 to strictly less than $(1-B)$ in expectation.
\end{theorem}
\begin{proof}
We prove this claim by considering the uniform additive simultaneous auction on $m$ items, where $m$ is even. The adversary chooses a subset $S\subseteq I$ of the items, with $|S| =\frac{m}{2}$, uniformly from the subsets of this size. He then bids $\frac{w_1(B)}{m}$ on each element in $S$, and $\frac{w_2(B)}{m}$ on each element not in $S$, where $w_1$ and $w_2$ are as follows.
\begin{align*}
    w_1(B) &= \twopartdef{2B}{0<B<\frac{1}{4}}{\frac{1}{3}+\frac{2B}{3}}{\frac{1}{4}\leq B<1} &\\
    w_2(B) &= \twopartdef{0}{0<B<\frac{1}{4}}{\frac{4B}{3}-\frac{1}{3}}{\frac{1}{4}\leq B<1}
\end{align*}

It is easily shown that this strategy is feasible for Bidder~2, and that Bidder~1's best response is to bid $\frac{w_1(B)}{m}$ on every item and win all the items. Bidder~1's profit is then
\begin{equation*}
    \pi_1^* = \twopartdef{1-2B}{0<B<\frac{1}{4}}{\frac{2}{3}(1-B)}{\frac{1}{4}\leq B<1}
\end{equation*}
which is strictly smaller than $(1-B)$ when $0<B<1$.
\end{proof}

We remark that the strategies used 
in proving Theorems~\ref{thm:simul-second} and~\ref{thm:simul-first} require no knowledge of the adversary's budget. Bidder~$1$ can implement them 
based solely on her own valuation function so these profit guarantees are extremely robust.

So, indeed, the adversary is weaker in a simultaneous auction than in the corresponding sequential auction (for example, the bound of Theorem~\ref{thm:simul-first} is 
larger than that
of Theorem~\ref{thm:xos-lb} for $B> 0.18$).
In addition, unlike for sequential auctions, the power of the adversary differs in a simultaneous 
auction depending on whether a first-price or second-price mechanism is used: the adversary is stronger in a first-price auction. Finally, unlike the sequential case, it is essential to introduce randomization to obtain non-trivial bounds in the first-price simultaneous setting.

\section{Bounds for Subadditive Valuation Functions} \label{s:subadd-bounds}

In this section we make a return to sequential auctions. We study the risk-free profitability of Bidder 1 when her valuation function is subadditive. 
Since there exist subadditive functions that are not XOS, the simple strategy from Section \ref{ss:xos-lb} 
is no longer guaranteed to work. Indeed, we present in 
Section \ref{ss:subadd-ub} a class of examples of subadditive valuations whose risk-free profitability is 
strictly less than $f(B) = (1-\sqrt{B})^2$ for an adversary with budget $B < \frac{1}{4}$. The relationship between the classes of XOS functions and subadditive functions was explored by \citet{BR11}, 
via the class of {\em $\beta$-fractionally subadditive} valuation functions.
\begin{prop}\label{prop:xos-sa}\cite{BR11}
Every subadditive valuation is $\ln{m}$-fractionally subadditive.
\end{prop}

It follows from this proposition that there exists a bid vector $r$ satisfying $\sum_{j:a_j\in I}r_j = v(I)$ and $\sum_{j:a_j\in S}r_j \leq \ln{m}\cdot v(S)$ for each 
subset $S \subseteq I$. In \cite{BR11}, the 
authors also provide an example showing that this is tight. Consequently, if Bidder~1 plays a strategy analogous to the strategy from Section \ref{ss:xos-lb} on this example, using the bid vector $r$ in place of the additive function $\gamma^*$, then any strict subset $S$ of $I$ that Bidder~1 wins is only guaranteed to have value $O(\frac{1}{\ln{m}})$ and, potentially, this guarantees a profit of only $O(\frac{1}{\ln{m}})$. This relationship indicates an inherent difficulty in showing a non-trivial lower bound on the profitability of subadditive valuations. However, we make progress on an important special case, namely subadditive valuations on identical items. 
Here, every subset $S$ of $I$ such that $|S| = k$ where $0 \leq k \leq m$ 
has the same value that we denote by $v(k)$. The earlier assumptions still hold, so $v(0) = 0$, and $v$ is monotone. In Section~\ref{ss:subadd-lb}, we present a strategy for Bidder~1 that gives a new lower bound on the profitability. Then, in Section~\ref{ss:subadd-ub}, we prove that this lower bound is tight when the budget $B$ is in $(0,\frac{1}{4})$. Moreover, the lower bound is tight at every $B$ of the form $(\frac{k}{k+1})^2$ 
for any positive integer $k$, and we conjecture that this tightness extends to all $B\in(0,1)$.

\subsection{The Subadditive Lower Bound with Identical Items} \label{ss:subadd-lb}
We obtain our lower bound on the profitability of Bidder~1 with a simple strategy: Bidder~1 chooses a constant price $\tilde{p}$ and a 
target allocation $\tilde{q}$ in advance, and bids $\tilde{p}$ on every item, stopping when she wins $\tilde{q}$ items. We will need the following claim.
\begin{claim} \label{clm:subadd1}
For any set $S \subseteq I$, where $|S| = q$, $v(S) \geq \frac{v(I)}{\ceil{\frac{m}{q}}}$.
\end{claim}

\begin{proof}
    Let $S$ be a subset of $I$ of size $q$. We want to show that $v(S) = v(q) \geq \frac{v(I)}{\ceil{\frac{m}{q}}}$. Consider any partition of the 
    set $I$ into $\ell = \ceil{\frac{m}{q}}$ sets $S_1,\ldots,S_\ell$, where $S_1 = S$, and each of the first $\floor{\frac{m}{q}}$ 
    sets $S_1,\ldots,S_{\floor{\frac{m}{q}}}$ has size $q$. 
    
    We have two cases. If $m$ is a multiple of $q$, then $\ceil{\frac{m}{q}} = \floor{\frac{m}{q}} = \ell$. By subadditivity, we have
    \begin{equation*}        
    v(I) \ = \ v(S_1 \cup \ldots \cup S_\ell)
        \ \leq\  v(S_1) + \ldots + v(S_\ell)
       \  = \ \ell\cdot v(q)
    \end{equation*}  
    Thus we have $v(q) \geq \frac{v(I)}{\ell} = \frac{v(I)}{\ceil{\frac{m}{q}}}$.
    
    Now if $m$ is not a multiple of $q$, then $\ceil{\frac{m}{q}} - \floor{\frac{m}{q}} = 1$. Then $|S_\ell| = m - q\cdot\floor{\frac{m}{q}} < q$, and we have
    \begin{align*}  
        v(I) &= v(S_1 \cup \ldots \cup S_\ell) & \\
        &\leq v(S_1) + \ldots + v(S_\ell) &\\
        &= (\ell-1)\cdot v(q) + v(m-q\cdot\floor{\frac{m}{q}}) &\\
        &\leq (\ell-1)\cdot v(q) + v(q) &\\
        &= \ell\cdot v(q).
    \end{align*}  
    Here the first in equality follows by subadditivity and the second by monotonicity.
    So we have $v(q) \geq \frac{v(I)}{\ell} = \frac{v(I)}{\ceil{\frac{m}{q}}}$ as required.
\end{proof}

Now, for an appropriate choice of $\tilde{p}$ and $\tilde{q}$, we show that Bidder~1 can guarantee a profit of at least $t^*(B) - O(\frac{1}{m})$, where
\begin{align*}
    t^*(B) = \max_{k \in \mathbb{Z}_{\geq 1}} t_k(B).
\end{align*}

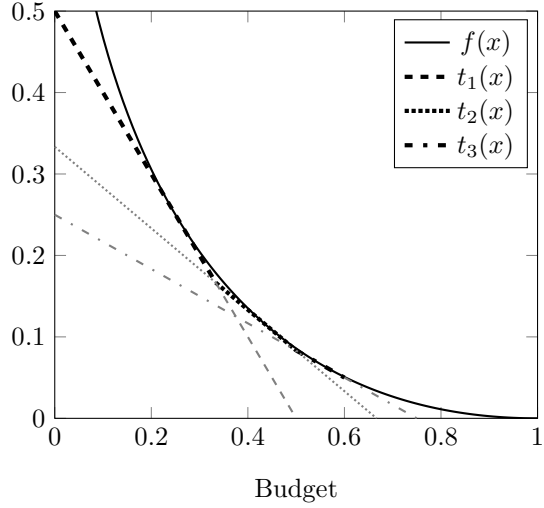
\begin{figure}[ht]
    \centering
    \begin{tikzpicture}
        \begin{axis}
        [
            width=8cm,
            height=7cm,
            title={Budget},
            title style=
            {
                at={(current axis.south)},
                anchor=north,
                outer sep=.9cm,
            },
            xmin=0,
            xmax=1,
            ymin=0,
            ymax=0.5,
            xticklabel style={
                /pgf/number format/fixed,
                /pgf/number format/precision=5
            },
            scaled x ticks=false,
            yticklabel style={
                /pgf/number format/fixed,
                /pgf/number format/precision=5
            },
            scaled y ticks=false,
            domain=0:1,
            samples=100,
        ]
        \addplot[thick]
            {1+x-2*sqrt(x)};
        \addplot[ultra thick, dashed]
            coordinates {
                (0,0.5)
                (0.333333,0.166667)
            };
        \addplot[ultra thick, densely dotted]
            coordinates {
                (0.333333,0.166667)
                (0.5,0.083333)
            };
        \addplot[ultra thick, loosely dashdotted]
            coordinates {
                (0.5,0.083333)
                (0.6,0.05)
            };

        \addplot[gray, thick, dashed]
            coordinates {
                (0.333333,0.166667)
                (0.5,0)
            };
        \addplot[gray, thick, densely dotted]
            coordinates {
                (0,0.333333)
                (0.333333,0.166667)
            };
        \addplot[gray, thick, densely dotted]
            coordinates {
                (0.5,0.0833333)
                (0.666667,0)
            };
        \addplot[gray, thick, loosely dashdotted]
            coordinates {
                (0,0.25)
                (0.5,0.083333)
            };
        \addplot[gray, thick, loosely dashdotted]
            coordinates {
                (0.6,0.05)
                (0.75,0)
            };
        
        \legend{$f(x)$,$t_1(x)$,$t_2(x)$,$t_3(x)$}
        \end{axis}
    \end{tikzpicture}
    \caption{Plot of the functions $f(x)$, $t_1(x)$, $t_2(x)$, and $t_3(x)$. $t^*(x)$ is the piecewise linear function shown by the bold line segments.}
    \label{fig:t*}
\end{figure}

Interestingly, $t_k(B) = \frac{1}{k+1}-\frac{B}{k}$ is the tangent to our earlier lower bound of $f(B)=(1-\sqrt{B})^2$ 
at $B = (\frac{k}{k+1})^2$ (see Figure \ref{fig:t*}). Denote by $SI_m$ the subadditive valuations on $m$ identical items. We show the following lower bound on the profitability in the subadditive case with identical items.
\begin{theorem}\label{thm:subadd-lb}
$\mathcal{P}(SI_m,B) \geq t^*(B) - O(\frac{1}{m})$.
\end{theorem}
\begin{proof}
Consider a normalized instance of 
the subadditive sequential auction on identical items where the adversary has normalized budget $B$. First, for any choice of 
target allocation $q$, we want to find the minimum price-per-item $p_q$ so that the adversary's budget is insufficient to stop 
Bidder~1 from winning $q$ items. In other words, we require the smallest possible price $p_q$ so that the adversary cannot win 
more than $m-q$ items. That is, we need to let $p_q(m-q+1) = B+\delta$ for some negligibly small $\delta$. We may ignore the 
negligible additive $\delta$ term and let $p_q(m-q+1) = B$, from which we obtain $p_q = \frac{B}{m-q+1}$. So, for any $q$ such 
that $0\leq q \leq m$, in order to win exactly $q$ items, Bidder~1 bids a price $p_q = \frac{B}{m-q+1}$ on every item and stops when 
she wins $q$ items. By Claim \ref{clm:subadd1}, the value of this set in the normalized auction is at least $\frac{1}{\ceil{\frac{m}{q}}}$.

Finally, we want to choose a target allocation $\tilde{q}$ that maximizes Bidder~1's profit under this constant-price strategy. For the 
purpose of our analysis, for a fixed $m$, any optimal choice of $q$ must be of the form $q = \ceil{\frac{m}{k}}$ for some 
positive integer $k\geq2$, which is the smallest choice of $q$ that has value at least $\frac{1}{k}$ according to our lower 
bound in Claim \ref{clm:subadd1}. Thus for any positive integer $k$, the profit obtained from winning exactly $q = \ceil{\frac{m}{k}}$ 
items with our constant-price strategy is at least
\begin{align*}
    v(q) - q\cdot p_q &\geq \frac{1}{\ceil{\frac{m}{q}}} - \frac{qB}{m-q+1} &\\
    &= \frac{1}{\ceil{\frac{m}{\ceil{\frac{m}{k}}}}} - \frac{\ceil{\frac{m}{k}}B}{m-\ceil{\frac{m}{k}}+1} &\\
    &\geq \frac{1}{\ceil{\frac{m}{(\frac{m}{k})}}} - \frac{(\frac{m}{k}+1)B}{m-(\frac{m}{k}+1)+1}.
\end{align*}
Therefore
\begin{align*}  
    v(q) - q\cdot p_q &\geq \frac{1}{k} - B\left(\frac{m+k}{m(k-1)}\right) &\\
    &= \frac{1}{k} - \frac{B}{(k-1)} - \frac{1}{m}\cdot \left(\frac{Bk}{k-1}\right) &\\
    &= t_{k-1}(B) - \frac{1}{m}\cdot \left(\frac{Bk}{k-1}\right)
\end{align*}
where $t_k(B) = \frac{1}{k+1} - \frac{B}{k}$ is the tangent to $f(B) = (1-\sqrt{B})^2$ at $B = (\frac{k}{k+1})^2$. The second term in the 
above expression is an additive $O(\frac{1}{m})$ factor that we subtract from our lower bound, but this factor goes to 0 as the number 
of items increases. Bidder~1's strategy is then to choose $\tilde q$ to maximize this value, which is equivalent to maximizing over the sequence of tangents. This completes the proof of the lower bound.
\end{proof}

\subsection{The Subadditive Upper Bound with Identical Items} \label{ss:subadd-ub}

In this section we present a matching upper bound for the range $0 < B < \frac{1}{4}$. For any $B$ in this range, we 
have that $\max_{k \in \mathbb{Z}_{\geq 1}} t_k(B) = t_1(B)$, so our lower bound is just $t_1(B) = \frac{1}{2} - B$.  We match this lower 
bound by constructing a sequence of valuation functions and corresponding strategies for the adversary that give us the following theorem.

\begin{theorem} \label{thm:sai-ub}
$\mathcal{P}(SI_m,B) \leq t^*(B) + O(\frac{1}{\sqrt{m}})$ when $B\in(0,\frac{1}{4})$ and $m$ is larger than some constant $m_0$ that depends only on $B$.
\end{theorem}
\begin{proof}
For every budget $B = x$ with $0 < x < \frac{1}{4}$, let $\sigma(x) = \frac{8x}{1-4x}$. It is easy to see that $\sigma(x)$ is a (well-defined) positive 
real number for every $x$ in $(0,\frac{1}{4})$. For every $x$ in this interval and every positive integer $m$, we define the normalized identical-item 
auction on $m$ items $\mathcal{S}_{x,m}$ in the following manner.  We set $v(1) = \frac{1}{2+\sigma(x)}$, and for every $i \in \{2,\ldots,m-1\}$ 
we set $v(i) = \frac{1}{2+\sigma(x)} + \frac{(i-1)\sigma(x)}{d(2+\sigma(x))}$, where $d = m-2$. Finally, we set $v(m) = 1$. In other words, we first 
construct an unnormalized valuation function by setting the marginal value of obtaining a first item and an $m^{\text{th}}$ item to 1, and the 
marginal value of getting an $i^{\text{th}}$ item to $\frac{\sigma(x)}{d}$ for $2 \leq i \leq m-1$. Since the total value of the $m$ items is $2+\sigma(x)$, 
we divide the value of every set by this factor to obtain the above normalized instance $\mathcal{S}_{x,m}$.

We have an important reason for choosing $\sigma(x) = \frac{8x}{1-4x}$: with this choice of $\sigma(x)$, the 
expression $\frac{\sigma(x)}{2+\sigma(x)}$, which shows up repeatedly in the following analysis, simplifies to the expression $4x$. Put differently, 
for any fixed $x$, we consider valuation functions where the marginal value of getting a first item or an $m^{\text{th}}$ is $\frac{1-4x}{2}$, and the total 
marginal value of getting an additional $m-2$ items having already been allocated one item is $4x$. The reader may have observed that with this 
choice of $\sigma(x)$, this valuation function is not always subadditive. This is true: the function is subadditive if and only if $\sigma(x) \leq d$. 
When $\sigma(x) \leq d$, for any integers $i,j$ with $i \geq 1$, $j \geq 1$, and $i+j \leq m$, it is easy to see 
that $v(i+j) \leq \frac{2}{2+\sigma(x)} + \frac{(i+j-2)\sigma(x)}{d(2+\sigma(x))}$, whereas $v(i) \geq \frac{1}{2+\sigma(x)} + \frac{(i-1)\sigma(x)}{d(2+\sigma(x))}$ 
and $v(j) \geq \frac{1}{2+\sigma(x)} + \frac{(j-1)\sigma(x)}{d(2+\sigma(x))}$, so we have $v(i+j) \leq v(i) + v(j)$ and the function is subadditive. For the other 
direction, if $\sigma(x) > d$ then $v(2) = \frac{1+\frac{\sigma(x)}{d}}{2+\sigma(x)} > \frac{2}{2+\sigma(x)} = 2v(1)$, so the function is not subadditive. 
Since $d = m-2$, we require $m \geq \sigma(x) + 2$ for subadditivity.

To find our upper bound, we will consider auction instances $\mathcal{S}_{x,m}$ where the adversary can play according to the following 
strategy. Initially, the adversary makes a bid of $0$ on every item until Bidder~1 wins an item. Then, after Bidder~1 wins her first item, if the 
adversary has not yet won an item he makes a bid of $\frac{1+\sigma(x)}{d(2+\sigma(x))}$ on every item until Bidder~1 loses an item. At this 
point, each of the two bidders has won at least one item, so any allocation in the remaining subgame gives at least 1 item and at most $m-1$ 
items to Bidder~1. Consequently, the remaining subgame is simply a scaled instance of a uniform additive auction $\mathcal{A}_{m'}$ for 
some $m' \leq d$, and from this point on Bidder~2 simply plays an optimal strategy for the uniform additive auction. The adversary's bid 
of $\frac{1+\sigma(x)}{d(2+\sigma(x))}$ on each item induces a lower bound on the size of the auction instances that we consider: we 
require $x \geq \frac{1+\sigma(x)}{d(2+\sigma(x))}$ for this bid to be feasible. Substituting $d = m-2$ into the above inequality, we 
have $m \geq \frac{1+\sigma(x)}{x(2+\sigma(x))}+2$. Combining this with the lower bound for subadditivity, for any fixed $x$ we 
let $\mathcal{L}(x) = \max(\sigma(x)+2,\frac{1+\sigma(x)}{x(2+\sigma(x))}+2)$, and we only consider auction instances $\mathcal{S}_{x,m}$ 
on $m \geq \mathcal{L}(x)$ items, since these are the only subadditive instances in which the above strategy is feasible. Observe that for 
any fixed $x \in (0,\frac{1}{4})$, $\mathcal{L}(x)$ is a constant. For any budget in this range, the following theorem provides an upper bound on 
the risk-free profit in a subadditive, identical-item sequential auction. This upper bound applies to the general subadditive case, 
showing that the profitability of subadditive functions is strictly less than that of additive, submodular or XOS functions.

Consider the auction instance $\mathcal{S}_{x,m}$ on $m$ identical items. Our goal is to show that if Bidder~2 plays according to the above 
strategy in this auction, then for every strategy of Bidder~1, her risk free profit matches the lower bound. In particular, Bidder~1's best response 
to this strategy has a payoff of at most $t_1(x) + O(\frac{1}{\sqrt{m}}$). Suppose Bidder~1 plays a strategy such that $j_1$ is the index of the first item 
won by Bidder~1 and $j_2 > j_1$ is the index of the first item won by Bidder~2 after Bidder~1 wins an item (where necessary, we will consider separately 
the case where Bidder~2 does not win any item after Bidder~1 wins item $a_{j_1}$). We will show that for any feasible choice of $j_1$ and $j_2$, 
Bidder~1's payoff in any strategy that results in this choice is at most $t_1(x) + O(\frac{1}{\sqrt{m}}) = \frac{1}{2} - x + O(\frac{1}{\sqrt{m}})$. We have the 
following cases for $j_1$ and $j_2$.\medskip

\noindent{\tt Case 1:} Bidder~2 wins the first item.\\ 
If Bidder~2 wins the first item, then $j_1 \geq 2$. Immediately after Bidder~1 wins item $a_{j_1}$ (for a price of 0), the remaining subgame is the 
uniform additive auction $\mathcal{A}_{m'}$ where $m' = m-j_1$ and the remaining unscaled value is $\frac{\sigma(x)}{2+\sigma(x)} - \frac{(j_1-2)\sigma(x)}{d(2+\sigma(x))}$. 
Since Bidder~2 has only made bids equal to $0$, the remaining unscaled budget is $x$. We first consider the case where Bidder~1 wins the second item, so $j_1 = 2$. 
There are $d = m-2$ items remaining, so the remaining value is simply $\frac{\sigma(x)}{2+\sigma(x)}$. Consequently the remaining subgame is an instance of the 
uniform additive auction $\mathcal{A}_{d}$ with scaled budget $\frac{2+\sigma(x)}{\sigma(x)}\cdot x$ when the value is scaled to $1$. Since Bidder~1 
makes a profit of $\frac{1}{2+\sigma(x)}$ on the item $a_{j_1}$, the risk-free profitability of any subgame with $j_1 = 2$ is
\begin{align*}
    \frac{1}{2+\sigma(x)} + \frac{\sigma(x)}{2+\sigma(x)}f_d\left(\frac{2+\sigma(x)}{\sigma(x)}\cdot x\right) &= \frac{1}{2+\sigma(x)} + 4xf_d\left(\frac{1}{4}\right) &\\
    &\leq \frac{1}{2+\sigma(x)} + 4x\left(f\left(\frac{1}{4}\right) + \frac{1}{\sqrt{d}}\right) &\\
    &\qquad\qquad \text{by Theorem \ref{thm:xos-ub}} &\\
    &= \frac{1}{2+\sigma(x)} + 4x\left(\frac{1}{4} + \frac{1}{\sqrt{d}}\right) &\\
    &= \left(\frac{1}{2+\sigma(x)} + x\right) + \frac{4x}{\sqrt{d}} &\\
    &= \left(\frac{1}{2} - x\right) + \frac{4x}{\sqrt{d}} &\\
    &= t_1(x) + O\left(\frac{1}{\sqrt{m}}\right)
\end{align*}

Observe that for any strategy where $j_1 > 2$, Bidder~1 is simply giving up items to the adversary at a price of 0, and the remaining 
subgame is a uniform additive auction where the unscaled budget of the adversary remains the same but the total number of 
items (and total value) decreases. Since $g_m(x,0) \geq h_m(x,0)$ (see Section \ref{ss:xos-ub}), the risk-free profitability of any strategy 
where $j_1 > 2$ is upper bounded by the risk-free profitability of a strategy with $j_1 = 2$, and this profitability is at 
most $t^*(x) + O\left(\frac{1}{\sqrt{m}}\right)$ for an adversary with budget $x$.\medskip

\noindent{\tt Case 2:} Bidder~1 wins the first $m-1$ items.\\
It remains to consider the cases where $j_1 = 1$, where Bidder~1 wins the first item. First, suppose Bidder~1 wins the first item (for a price of 0), 
and then wins the next $m-2$ items. After Bidder~1 wins the first item, since Bidder~2 has not won an item his strategy is to 
bid $\frac{1+\sigma(x)}{d(2+\sigma(x))}$ on each of the remaining items, so the total price paid by Bidder~1 for the next $d = m-2$ items is 
at least $\frac{1+\sigma(x)}{2+\sigma(x)}$. Hence Bidder~1's profit is at most $1 - \frac{1+\sigma(x)}{2+\sigma(x)} = \frac{1}{2+\sigma(x)} = \frac{1}{2} - 2x$, 
which is at most $t_1(x)$.\medskip

\noindent{\tt Case 3:} $j_1 = 1$ and $j_2 \leq m-1$.\\ 
The only remaining case is where Bidder~1 wins the first item for a price of 0, and Bidder~2 wins an item in $\{a_{2},\ldots,a_{m-1}\}$. After Bidder~1 
wins the first item, Bidder~2 bids $\frac{1+\sigma(x)}{d(2+\sigma(x))}$ on every item until he wins the item $a_{j_2}$. Bidder~1 pays a 
total of $(j_2-2)\frac{1+\sigma(x)}{d(2+\sigma(x))}$ for the first $j_2-1$ items, which have total value $\frac{1}{2+\sigma(x)}+(j_2-2)\frac{\sigma(x)}{d(2+\sigma(x))}$. 
Since the price of each of these items (which is $\frac{1+\sigma(x)}{d(2+\sigma(x))}$) is greater than the marginal value of each of these 
items (which is $\frac{\sigma(x)}{d(2+\sigma(x))}$), by choosing to increase $j_2$ Bidder~1 is winning items from the adversary at a price that is 
higher than their marginal value in the remaining uniform additive auction. Since $g_m(x,1) \leq h_m(x,1)$ (see Section \ref{ss:xos-ub}), the risk-free 
profitability of any strategy where $j_2 > 2$ is lower than the risk-free profitability of a strategy where $j_2 = 2$. Consequently, we will fix $j_2 = 2$ and 
show that Bidder~1 makes a profit of at most $t^*(x) + O\left(\frac{1}{\sqrt{m}}\right)$ with this choice.

Now, after Bidder~2 wins item $a_2$, the remaining subgame is an instance of the uniform additive auction on $m-2$ items where the total 
unscaled value is $\frac{\sigma(x)}{2+\sigma(x)}$ which is equal to $4x$, and the total unscaled budget is $x - \frac{1+\sigma(x)}{d(2+\sigma(x))}$, 
which simplifies to $\frac{dx-2x-\frac{1}{2}}{d}$. When the budget is scaled by $\frac{1}{4x}$, it becomes $\frac{dx-2x-\frac{1}{2}}{4dx}$. Finally, we 
also need to add the profit made by Bidder~1 from item $a_1$, which is $\frac{1}{2+\sigma(x)}$. Putting all this together, Bidder~1's profit is at most
\begin{align*}
    \frac{1}{2+\sigma(x)} &+ 4xf_{m-2}\left(\frac{dx-2x-\frac{1}{2}}{4dx}\right) &\\
    &= \frac{1-4x}{2} + 4x\left[1 + \frac{dx-2x-\frac{1}{2}}{4dx} - 2\sqrt{\frac{dx-2x-\frac{1}{2}}{4dx}} + \frac{1}{\sqrt{m-2}}\right] &\\
    &\qquad\qquad\qquad\qquad \text{by Theorem \ref{thm:xos-ub}} &\\
    &= \frac{1}{2} - 2x + 4x + \frac{dx-2x-\frac{1}{2}}{d} - 2\sqrt{\frac{4x(dx-2x-\frac{1}{2})}{d}} + \frac{4x}{\sqrt{m-2}} &\\
    &= \frac{1}{2} + 3x - \frac{(2x+\frac{1}{2})}{d} - 2\sqrt{4x^2 - \frac{4x(2x+\frac{1}{2})}{d}} + \frac{4x}{\sqrt{m-2}} &\\
    &= \frac{1}{2} + 3x - \frac{(2x+\frac{1}{2})}{d} - 2\sqrt{4x^2 \left[1 - \frac{2x+\frac{1}{2}}{dx}\right]} + \frac{4x}{\sqrt{m-2}} &\\
    &= \frac{1}{2} + 3x - 4x\sqrt{1 - \frac{2x+\frac{1}{2}}{dx}} - \frac{(2x+\frac{1}{2})}{d} + \frac{4x}{\sqrt{m-2}} 
\end{align*}
Since $m \geq \mathcal{L}(x)$, $d \geq \frac{1+\sigma(x)}{x(2+\sigma(x))}$, so $dx \geq \frac{1+\sigma(x)}{2+\sigma(x)} = 2x + \frac{1}{2}$, 
so we have that $0 \leq 1-\frac{2x+\frac{1}{2}}{dx} \leq 1$. Then
\begin{align*}
    \frac{1}{2} + 3x &- 4x\sqrt{1 - \frac{2x+\frac{1}{2}}{dx}} - \frac{(2x+\frac{1}{2})}{d} + \frac{4x}{\sqrt{m-2}} &\\
    &\leq \frac{1}{2} + 3x - 4x\left(1 - \frac{2x+\frac{1}{2}}{dx}\right) - \frac{(2x+\frac{1}{2})}{d} + \frac{4x}{\sqrt{m-2}} &\\
    &= \frac{1}{2} - x + \left[\frac{6x+\frac{3}{2}}{d} + \frac{4x}{\sqrt{m-2}}\right] &\\
    &= t_1(x) + O(\frac{1}{\sqrt{m}}).
\end{align*}
This completes the proof of Theorem \ref{thm:sai-ub}.
\end{proof}

An important consequence of the above result is that the lower bound for XOS valuations does not hold for subadditive valuations. 
This differentiates the class of subadditive valuations from the additive, submodular and XOS classes in that Bidder~1 can no 
longer guarantee a profit of $(1-\sqrt{B})^2$ when her valuation function is subadditive.

\bibliography{resources}

\end{document}